\documentclass[aps,prl,reprint,superscriptaddress,floatfix,longbibliography]%
{revtex4-2}
\usepackage[utf8]{inputenc}
\usepackage[T1]{fontenc}

\usepackage{amsmath,amssymb,amsthm,mathtools}
\usepackage{bm}
\usepackage{microtype}
\usepackage[hidelinks]{hyperref}
\usepackage{graphicx}
\usepackage{helvet}
\usepackage{tikz}
\usetikzlibrary{arrows.meta,calc,positioning}
\usepackage{capt-of}
\usepackage{tabularray}
\usepackage{newtxtext}
\usepackage{enumitem}

\usepackage[dvipsnames]{xcolor}

\makeatletter
\renewcommand\section{\@startsection{section}{1}{\z@}
 {1.25ex \@plus 1ex \@minus .2ex}
 {-0.8em}
 {\normalfont\itshape}}
\renewcommand\subsection{\@startsection{subsection}{2}{\z@}
 {1ex \@plus .8ex \@minus .2ex}
 {-0.8em}
 {\normalfont\itshape}}
\makeatother

\newtheorem{theorem}{Theorem}
\newtheorem*{theorem*}{Theorem}

\newtheorem{reptheorem}{Theorem}



\newcommand{\ii}{\mathrm{i}}
\newcommand{\xx}{i}
\newcommand{\yy}{j}
\newcommand{\re}{\operatorname{Re}}
\newcommand{\im}{\operatorname{Im}}
\newcommand{\Aff}{\mathcal F}
\newcommand{\W}{R}
\newcommand{\maxrate}{r}
\newcommand{\winding}{\omega}
\newcommand{\windA}{\winding}
\newcommand{\windFunc}{\omega}

\newcommand{\pp}{\bm{p}}

\newcommand{\mixc}{\tau}
\newcommand{\xlog}{\xi}
\newcommand{\phase}{\theta}
\newcommand{\branchParam}{\phi}
\newcommand{\fwd}{k^+}
\newcommand{\bwd}{k^-}

\newcommand{\Arg}{\operatorname{Arg}}

\newcommand{\imL}{\im\lambda}
\newcommand\absImL{\vert\imL\vert}

\newcommand{\reL}{\re\lambda}
\newcommand{\minusReL}{-\reL}
\newcommand{\dualReL}{2\maxrate+\reL}

\newcommand{\RR}{\mathbb{R}}

\newcommand\CCC{\mathcal{C}}
\newcommand\CCCpos{\mathcal{C}^+}
\newcommand\yyxx{\yy\xx}
\newcommand\xxyy{\xx\yy}
\newcommand\yxedge{(\yy\leftarrow \xx)}

\newcommand\Conv{\mathsf{Conv}}

\newcommand\edgeweight{g}
\newcommand\oyxedge{\edgeweight_{\yy \leftarrow \xx}}

\hypersetup{pdftitle={Cycle affinity and winding localize eigenvalues of Markov
generators},
  pdfauthor={Artemy Kolchinsky, Naruo Ohga, Sosuke Ito}}

\begin{document}

\title{Cycle affinity and winding localize eigenvalues of Markov generators}

\author{Artemy Kolchinsky}
\email{artemyk@gmail.com}
\affiliation{ICREA-Complex Systems Lab, Universitat Pompeu Fabra, 08003
Barcelona, Spain}
\affiliation{Universal Biology Institute, Graduate School of Science, The
University of Tokyo, 7-3-1 Hongo, Bunkyo-ku, Tokyo 113-0033, Japan}
\affiliation{Barcelona Collaboratorium, Wellington 30, 08005 Barcelona, Spain}

\author{Naruo Ohga}
\affiliation{Department of Physics, Graduate School of Science, Kyoto
University, Kyoto 606-8502, Japan}
\affiliation{Universal Biology Institute, Graduate School of Science, The
University of Tokyo, 7-3-1 Hongo, Bunkyo-ku, Tokyo 113-0033, Japan}

\author{Sosuke Ito}
\affiliation{Department of Physics, Graduate School of Science, The University
of Tokyo, 7-3-1 Hongo, Bunkyo-ku, Tokyo 113-0033, Japan}
\affiliation{Universal Biology Institute, Graduate School of Science, The
University of Tokyo, 7-3-1 Hongo, Bunkyo-ku, Tokyo 113-0033, Japan}

\begin{abstract}
The complex eigenvalues of Markov generators govern oscillatory properties of
relaxation, autocorrelation, and linear response. We show that these eigenvalues
are localized by nonequilibrium cycles of the generator, thus revealing a
fundamental tradeoff between thermodynamic driving, oscillation, and decay of
eigenmodes. Specifically, we prove that each complex eigenvalue is confined to a
region determined by the cycle affinity and the eigenvector ``winding number''
of some nonequilibrium cycle. In discrete and continuous unicyclic systems, we
also demonstrate that the winding number is determined by the ordered eigenvalue
index, yielding new thermodynamic bounds on the slowest and fastest relaxation
modes. In discrete multicyclic systems, our approach unifies and extends several
previous inequalities and proves the Uhl--Seifert ellipse conjecture.
\end{abstract}

\maketitle

Markov chains are used to study various nonequilibrium systems, ranging from
chemical reaction
networks~\cite{van1992stochastic,schnakenbergNetworkTheoryMicroscopic1976} and
biomolecular
machines~\cite{qian2003thermodynamic,skoge2013chemical,MehtaSchwab2012}
 to population dynamics~\cite{allen2010introduction} and social
processes~\cite{liggett2013stochastic}.
In a continuous-time Markov chain, the
probability distribution $\pp(t)$ evolves according to the master equation $d_t
\pp(t) = \W \pp(t)$,
where $\W$ is the \emph{generator} (also called the \emph{rate matrix}).

The
eigenvalues and eigenvectors of the generator govern many important dynamical
properties.  Specifically, each nonzero eigenvalue $\lambda$  corresponds to
an eigenmode with decay rate $-\reL$ and oscillation frequency
$\vert \imL\vert$, and it governs oscillatory properties of relaxation
dynamics~\cite{BaratoSeifert2017}, autocorrelation
functions~\cite{OhgaItoKolchinsky2023}, and linear
response~\cite{hanggi1982stochastic}.  For this reason, the localization of
eigenvalues is relevant to the functional performance of systems like molecular
clocks,
oscillators, and
sensors~\cite{BaratoSeifert2017,DelJuncoVaikuntanathan2020,%
OberreiterSeifertBarato2022,ZhengTang2024,UhlSeifert2019,shiraishi2023entropy,%
KolchinskyOhgaIto2024}.

A Markov generator can also be characterized by whether it satisfies the
principle of \emph{detailed balance} (DB). An irreducible generator satisfies DB
if and only if the \emph{cycle affinity} (defined below) vanishes for every
cycle~\cite{schnakenbergNetworkTheoryMicroscopic1976,Seifert2012} or
equivalently if the generator has an equilibrium steady state in which all net
fluxes vanish. It is known that DB generators have only real-valued
eigenvalues~\cite{wei_structure_1962,hearon1953kinetics}, so their eigenmodes
cannot exhibit   stochastic oscillations or resonant response to periodic
driving.

Conversely, DB is violated if at least one cycle has nonvanishing affinity,  or
equivalently if the generator has a nonequilibrium steady state with
nonvanishing fluxes. At the microscopic level, the breaking of DB results from
thermodynamic driving, such as coupling to nonequilibrium chemical potentials,
temperature differences, stirring,
etc.~\cite{schnakenbergNetworkTheoryMicroscopic1976,Seifert2012}. The cycle
affinity measures the total thermodynamic force accumulated around a cycle, and
it serves as a quantitative measure of the driving strength.
Unlike the steady-state entropy production rate --- another common measure of
nonequilibrium --- the cycle affinity is determined by thermodynamic control
parameters (reservoir chemical potentials, temperatures, etc.) and does not
depend on microscopic kinetics or the stationary
distribution~\cite{OhgaItoKolchinsky2023}.

In this Letter, we demonstrate a fundamental relationship between thermodynamic
driving, graph topology, and eigenvalue localization.
Specifically, we show that each nonreal eigenvalue is confined to a region
determined by the thermodynamic cycle affinity and the eigenvector \emph{winding
number} on some nonequilibrium cycle of the generator. The winding number is a
topological invariant that quantifies the eigenvector's spatial frequency. In
unicyclic systems, we further prove that this winding number is determined by
the ordered eigenvalue index, yielding new thermodynamic constraints on the
slowest and fastest relaxation modes.
A brief illustration of our results on a model of a biochemical clock is
provided in the \emph{End Matter}.

Our results contribute to a growing body of literature on the relationship
between spectral and thermodynamic properties of Markovian
systems~\cite{BaratoSeifert2017,NguyenSeifertBarato2018,UhlSeifert2019,%
DelJuncoVaikuntanathan2020,RemleinWeissmannSeifert2022,%
OberreiterSeifertBarato2022,OhgaItoKolchinsky2023,shiraishi2023entropy,%
KolchinskyOhgaIto2024,van2024dissipation,ZhengTang2024,%
XuKolchinskyDelvenneIto2025,SantolinFalasco2025,NagayamaIto2025,%
Gu2026DissipationCoherence,NagayamaKolchinskyIto2026,%
Kolchinsky2026DissipationCoherence}, and they provide a unified route to derive
various existing bounds. In particular, we derive a sharpened version of our
previous eigenvalue
sector bound~\cite{OhgaItoKolchinsky2023} (itself inspired by the
Barato--Seifert conjecture~\cite{BaratoSeifert2017}), resolve the Uhl--Seifert
ellipse conjecture~\cite{UhlSeifert2019}, and uncover several novel eigenvalue
inequalities.
Many of our results remain nontrivial even in the regime of absolute
irreversibility, where they provide refinements of classical inequalities such
as the Dmitriev--Dynkin bound~\cite{DmitrievDynkin1946} and
the Kellogg--Stephens bound~\cite{kellogg1978complex}.

In the following, we state our results for continuous-time generators. They can
be adapted to discrete-time transition matrices via the standard embedding
$T=I+\W/(\max_\xx (-\W_{\xx\xx}))$, which rescales and
shifts the eigenvalues along the real
axis.

\newcommand\BndAff{\Aff}
\section{Unicyclic systems.---}
We begin with the simple case of unicyclic systems, often used to model
biomolecular devices such as
oscillators~\cite{BaratoSeifert2017,OberreiterSeifertBarato2022,%
wierenga2018quantifying,wachtel2018thermodynamically},
sensors~\cite{skoge2013chemical,qian2003thermodynamic,MehtaSchwab2012,%
KolchinskyOhgaIto2024}, and
motors~\cite{KolomeiskyFisher2007,seifert2011efficiency,%
PineroSoleKolchinsky2024}.
We consider an irreducible unicyclic generator $\W$ with $n \ge 3$ states,
\begin{align}
\W_{\yyxx}=\delta_{\yy,\xx+1}\fwd_\xx + \delta_{\yy,\xx-1}\bwd_\xx -
\delta_{\yy,\xx} (\fwd_\xx + \bwd_\xx)\,,
\label{eq:unicyclicW}
\end{align}
where $\delta$ is the Kronecker delta and indices are taken modulo $n$.
$\W$ is called a \emph{uniform cycle} if the forward rates are equal across
states, as are the backward rates ($\fwd_\xx \equiv \fwd,\bwd_\xx\equiv \bwd$).
We write $\maxrate:=\max_\xx -\W_{\xx\xx}$
for the maximum escape rate.

The \emph{cycle affinity} is the sum of the log-ratios of the forward and
backward transition rates: $\Aff:=\sum_{\xx=1}^n \ln({\fwd_\xx}/{\bwd_\xx})$,
with $\Aff=\infty$ if any transition is irreversible ($\bwd_\xx=0$). Without
loss of generality, we assume that states are ordered so that $\Aff \ge 0$.
Assuming local detailed balance~\cite{maes_ldb_notes},
the cycle affinity is equal to the increase of the thermodynamic entropy of the
environment
per cycle completion.

Let $\lambda$ be a nonreal eigenvalue of $\W$ with an associated (left or right)
eigenvector $\bm u$. A straightforward application of Gershgorin's Circle
Theorem~\cite{varga_gersgorin_2004} implies that $\vert \lambda+\maxrate \vert
\le \maxrate$; therefore, any nonreal eigenvalue obeys
 $-2\maxrate < \reL< 0$.
A key quantity in our results is the \emph{winding number} associated with the
eigenvector. For a unicyclic system, we define it as
\begin{align}
\windFunc(\bm u) :=\bigg\vert\frac{1}{2\pi}
\sum_{\xx=1}^n\Arg\frac{u_{\xx+1}}{u_\xx} \bigg\vert_n,
\label{eq:winding-unicycle}
\end{align}
where $\vert m\vert_n :=\min\{ m \bmod n,(n-m) \bmod n \}$
is the wrapped frequency and $\Arg(x) \in (-\pi,\pi]$ is the principal argument.
$\Arg(u_{i+1}/u_i)$ measures the increment of complex phase from $u_i$ to
$u_{i+1}$, and summing it around the cycle gives the total phase increment, an
integer multiple of $2\pi$. Consequently, the winding number is a quantized
topological invariant that counts how many times the eigenvector's phase rotates
as it goes around the cycle, quantifying the eigenvector's spatial frequency. In
a uniform cycle where the eigenvectors correspond to spatial Fourier modes,
\(u_\xx = e^{\ii 2\pi m\xx/n}\), the winding number is equal to the wrapped
spatial frequency, \(\winding = |m|_n\).

In SM\ref*{app:unicyclic-proofs} in the Supplemental
Material~\cite{SupplementalMaterial}, we prove that the winding number is
well-defined (because $u_\xx \ne 0$ for all $\xx$).
Theorem~\ref{thm:ordered-winding} also shows that it can be equivalently
computed using either the left or right eigenvector.
Where clear from context, we often write the winding number as $\windFunc$,
leaving dependence on the eigenvector implicit.

\begin{figure}
\centering
\includegraphics[width=\columnwidth]{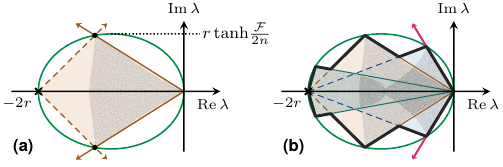}
\caption{Illustration of Theorem~\ref{thm:unicyclic-bounds} for a unicyclic
generator with $n=7$ states and affinity $\Aff= 15$. (a) Each nonreal eigenvalue
$\lambda$ obeys the bounds~\eqref{eq:unicyclic-primal} (solid brown)
and~\eqref{eq:unicyclic-dual} (dashed brown), which depend on the winding number
$\winding$ (here we show $\winding=2$).
(b) The union across all $\winding\in \{1,\dots,\lceil n/2\rceil - 1\}$ regions
gives a non-convex localization region for all eigenvalues (black outline). The
green ellipse is the bound corresponding to the Uhl--Seifert
conjecture~\cite{UhlSeifert2019}; our previous sector
bound~\cite{OhgaItoKolchinsky2023} is shown using magenta arrows.
Gray regions show numerical eigenvalues from sampled unicyclic generators with
$n=7,\Aff=15$~\cite{sampling-methods}.}
\label{fig:unicyclic}
\end{figure}

We now present our main result for unicyclic systems. Its derivation is deferred
to the \emph{End Matter}.

\begin{theorem}
\label{thm:unicyclic-bounds}
Given an irreducible unicyclic generator with cycle affinity $\Aff \ge 0$, every
nonreal eigenpair $(\lambda,\bm u)$ has winding number $\winding
\in\{1,\dots,\lceil n/2\rceil - 1\}$ and obeys
\begin{align}
\frac{\absImL }{\minusReL}
&\le
\tanh\frac{ \BndAff }{2n} \cot\frac{\pi \winding}{n},
\label{eq:unicyclic-primal}
\\
\frac{\absImL}{\dualReL}
&\le
\tanh\frac{ \BndAff }{2n} \tan\frac{\pi \winding }{n}.
\label{eq:unicyclic-dual}
\end{align}
Equality in either bound holds if and only if the cycle is uniform.%
\end{theorem}

The statement of the theorem is restricted to nonreal eigenpairs because the
winding number is not generally well-defined for real modes. Nonetheless, the
resulting localization regions include all real eigenvalues, for which
$\absImL=0$.

The theorem is illustrated in Fig.~\ref{fig:unicyclic}. Each eigenvalue
$\lambda$ is constrained by two
inequalities~\eqref{eq:unicyclic-primal}--\eqref{eq:unicyclic-dual}.
The first inequality implies that eigenmodes with fast oscillation, slow decay,
and high spatial frequency require strong driving (large affinity).
The second inequality implies that eigenmodes with fast oscillation, fast decay,
and low spatial frequency also require strong driving. Thus, fast oscillations
are incompatible with both slow and fast decay, and the winding number
$\winding$ (spatial frequency) determines the balance between these two
constraints.
The two bounds intersect at the uniform $n$-cycle eigenvalue with the same
winding number $\winding$, escape rate $\maxrate$ and affinity $\Aff$. Thus, the
uniform cycle requires minimal affinity to sustain an eigenmode with a given
winding number and eigenvalue location.

To test Theorem~\ref{thm:unicyclic-bounds} against numerical data, we computed
eigenvalues of unicyclic generators using an adaptive
search~\cite{sampling-methods}.  The results are shown as gray regions in
Fig.~\ref{fig:unicyclic}(a) for $\winding=2$ modes and in
Fig.~\ref{fig:unicyclic}(b) for all modes.  We verified that our inequalities
are always satisfied. We also observe significant empty regions, especially in
the winding-resolved Fig.~\ref{fig:unicyclic}(a), suggesting that it may be
possible to derive tighter localization bounds.

Theorem~\ref{thm:unicyclic-bounds} implies various downstream bounds,
illustrated in Fig.~\ref{fig:unicyclic}(b). For example, we may obtain a
localization region for all eigenvalues, independent of winding number, by
taking the union across possible winding numbers. Another implication is the
inequality $-\absImL/\reL\le \tanh(\BndAff/2n)\cot(\pi/n)$, as follows
from~\eqref{eq:unicyclic-primal} and the monotonicity of $\cot$. This bound was
first conjectured for the second eigenvalue by Barato and
Seifert~\cite{BaratoSeifert2017} and recently proved by the present authors in
Ref.~\cite{OhgaItoKolchinsky2023} using a different technique (isoperimetric
inequality). As a final example, the geometric mean
of~\eqref{eq:unicyclic-primal}--\eqref{eq:unicyclic-dual} gives
$\absImL/\sqrt{\minusReL\,(\dualReL)}\le \tanh(\BndAff/2n)$. This implies that
all eigenvalues belong to an ellipse of semi-width $\maxrate$ and semi-height
$\maxrate \tanh({\BndAff}/{2n})$. This ellipse bound was conjectured by Uhl and
Seifert in 2019~\cite{UhlSeifert2019} but has remained unproven until now.

As mentioned, the winding number $\winding$ can be understood as the
eigenvector's spatial frequency.
Interestingly, in~SM\ref*{app:unicyclic-proofs}, we also prove that it is
determined by the eigenvalue's index when sorted by descending real part.  Thus,
the unicyclic winding number may be interpreted in two different ways: either as
the eigenmode's spatial frequency, or as the eigenmode's index of decay.

\newcommand\orderedWindingStmt{Given an irreducible unicyclic generator,
sort the eigenvalues by real part as
$0=\lambda_1\ge \re\lambda_2\ge \dots \ge \re \lambda_n$.
Then, the left and right eigenvectors $\bm u,\bm v$ of every nonreal $\lambda_k$
satisfy
$\windFunc(\bm u)=\windFunc(\bm v)=\lfloor k/2\rfloor$.
}

\begin{theorem}
\label{thm:ordered-winding}
\orderedWindingStmt{}
\end{theorem}

For example, if the slowest relaxation mode (corresponding to the spectral gap)
is nonreal, it has the smallest winding number $\winding=1$. Similarly, if the
fastest relaxation mode is nonreal, it has the largest winding number
$\winding=\lceil n/2\rceil -1$. Importantly, the relationship between winding
number and eigenvalue index allows us to apply
Theorem~\ref{thm:unicyclic-bounds} without requiring detailed information about
the eigenvectors.

The physical intuition behind this theorem is that, within the class of
unicyclic systems, the winding number is topologically protected and locked to
the eigenvalue index. The result can also be interpreted as the discrete,
cyclic, and non-self-adjoint analogue of the Sturm oscillation
theorem~\cite{shubin_invitation_2020,bohner_oscillation_2003,%
angenent1988dynamics}, which relates eigenvalue indices to spatial oscillations
of eigenfunctions.

\section{Size-independent bounds and continuous systems.---}
Theorem~\ref{thm:unicyclic-bounds}
implies several size-independent bounds for unicyclic systems.
For example, inequality~\eqref{eq:unicyclic-primal} can be weakened as
\begin{align}
\frac{\absImL}{\minusReL}
&\le
\frac{\BndAff}{2\pi \winding}\le
\frac{\BndAff}{2\pi}\,.
\label{eq:unicyclic-no-n-bound}
\end{align}
The first bound follows from $\tanh (\BndAff/2n) \le \BndAff/2n$ and $\cot (\pi
\winding/n) \le n/(\pi \winding)$, which become tight when $n\to \infty$. The
second bound follows from $\winding \ge 1$.

These inequalities can be extended to continuous systems on a cycle, such as
driven diffusions in a periodic potential~\cite{Reimann2002,Seifert2012}. In
SM\ref*{app:fp}, we consider an overdamped Fokker--Planck equation
with smooth periodic drift $f(x)$ and diffusion $Q(x)>0$ fields, $\partial_t
\rho
=
-\partial_x[f \rho -\partial_x(Q \rho )]$, oriented to have nonnegative cycle
affinity $\Aff=\int_0^1 (f(x)/Q(x))\,dx\ge 0$.
We prove that every nonreal eigenpair $(\lambda,u)$ satisfies the
inequalities~\eqref{eq:unicyclic-no-n-bound}, where $\winding=\vert\int_0^1
\im({u'}/{u}) \,dx|/2\pi$ is the eigenfunction winding number. The first bound
is tight for uniform drift and diffusion fields. We also show that the winding
number is determined by the ordered eigenvalue index, in  analogy to
Theorem~\ref{thm:ordered-winding}.

On the other hand, the inequality~\eqref{eq:unicyclic-dual} implies
\begin{align}
\frac{\absImL}{\dualReL}
&\le
\frac{\BndAff}{\pi(n-2\winding)}\le
\begin{cases}{\BndAff}/{2\pi}& n\;\text{even}\\
{\BndAff}/{\pi}& n\;\text{odd}
\end{cases}\,.
\label{eq:unicyclic-no-n-bound-dual}
\end{align}
The first bound uses $\tanh (\BndAff/2n) \le \BndAff/2n$ and $\tan (\pi
\winding/n) \le 1/[\pi(1/2- \winding/n)]$, which become tight when $ n\to
\infty,\winding/n\to 1/2$. The second bound uses $\winding \le \lceil n/2\rceil
- 1$. The final bounds in Eq.~\eqref{eq:unicyclic-no-n-bound-dual} become
trivial in the continuous diffusive limit: the left side vanishes as the maximum
escape rate diverges, while the right side remains finite.

\section{Multicyclic systems.---}
We now consider an arbitrary (multicyclic) rate matrix $\W$ with $n$ states and
maximum escape rate
$\maxrate:=\max_\xx(-\W_{\xx\xx})$.
The term \emph{cycle} $c$ refers to a cyclic sequence of $n_c \ge 3$ distinct
states $(\xx_1,\dots,\xx_{n_c},\xx_1)$ connected by nonzero edges
$\W_{\yyxx}>0$.
We use $\CCC$ to indicate the set of all cycles of $\W$. For any cycle $c$, the
\emph{cycle affinity} is defined as  $ \Aff_c:=\sum_{\yxedge\in
c}\ln({\W_{\yyxx}}/{\W_{\xxyy}})$,
with $\Aff_c:=\infty$ if $c$ contains an irreversible edge ($\W_{\xxyy}=0$).

Let $\lambda$ be a nonreal eigenvalue of $\W$ with an associated left or right
eigenvector $\bm u$.
We say that a cycle $c$ is \emph{admissible} for $\bm u$ if $u_\xx \ne 0$ for
all states $\xx$ in the cycle and, for all edges $\yxedge \in c$, the phases of
the ratios $u_\yy/u_\xx$ belong to a common open interval of length $\pi$.
For an admissible cycle $c$, we define the eigenvector \emph{winding number} as
\begin{align}
\windFunc_c(\bm u)
:=
\bigg\vert\frac{1}{2\pi}\sum_{\yxedge\in c}\arg_{\branchParam}
\frac{u_\yy}{u_\xx}\bigg\vert_{n_c}\,.
\label{eq:winding-multicycle}
\end{align}
Here, $\vert m\vert_n :=\min\{ m\bmod n,(n-m) \bmod n \}$ is the wrapped
frequency
and $\arg_{\branchParam} z:=\Arg(e^{-\ii \branchParam}z)+\branchParam$
is the argument with branch cut at $\branchParam\pm \pi$. We choose the branch
parameter $\branchParam = \Arg[\sum_{\yxedge\in c} e^{\ii\Arg (u_\yy/u_\xx)}]$
as the argument of the average phase vector, which guarantees that the branch
cut avoids the  interval containing the edge-ratio phases.
As in the unicyclic case, we write $\windFunc_c$ when the eigenvector argument
is understood.

As in unicyclic systems, the multicyclic winding number is an integer that
counts how many times the eigenvector's phase rotates as it goes around the
cycle, thus it quantifies the spatial frequency of the eigenvector across the
cycle. The multicyclic definition is a bit more complicated than the unicyclic
one, Eq.~\eqref{eq:winding-unicycle}, due to the need to avoid branch cuts.
Nonetheless, when applied to a unicyclic system,
Eq.~\eqref{eq:winding-multicycle} agrees with Eq.~\eqref{eq:winding-unicycle}
(see \emph{End Matter}). Note that, unlike the unicyclic case, the winding
number does not have a simple relationship with the ordered eigenvalue index in
multicyclic systems.
 (See SM\ref*{app:winding-index-counterexample} for a specific counterexample.)
In addition, the multicyclic winding number can vary depending on the choice of
left and/or right eigenvector corresponding to an eigenvalue.

\newcommand{\ncdown}{n_c^\downarrow}
\newcommand{\snc}{s_{n_c}}

\newcommand{\boundtag}[1]{\mathsf{({#1})}}
\newcommand{\boundle}[1]{\mathrel{\mathop{\le}\limits^{\boundtag{#1}}}}
\newcommand{\tablefigureicon}[1]{
 \raisebox{-0.5\height}{\includegraphics[width=0.07\textwidth]{#1.pdf}}
}
\newcommand{\slowdecayicon}{\tablefigureicon{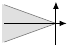}}
\newcommand{\ellipseicon}{\tablefigureicon{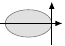}}
\newcommand{\heighticon}{\tablefigureicon{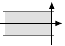}}
\newcommand{\fastdecayicon}{\tablefigureicon{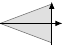}}

\newcommand\BndAffc{\Aff_c}

\newcommand\tAffc{\tanh\!\frac{\BndAffc}{2n_c}}
\renewcommand\tAffc{\mathfrak{F}_c}
\newcommand\tAffcMax{\max_{c\in\CCC}\tanh\!\frac{\BndAffc}{2n_c}}
\renewcommand\tAffcMax{\mathfrak{F}^*}
\begin{table*}[t!]
\caption{
Eigenvalue bounds implied by Theorem~\ref{thm:multicyclic-main}. Rows correspond
to different values of $\mixc$ (localization regions), columns correspond to
different types of information.
For convenience, we write \(\mathfrak F_c:=\tanh(\BndAffc/2n_c)\), \(\mathfrak
F^*:=\max_{c\in\CCC}\tanh(\BndAffc/2n_c)\), and \(n^\downarrow:=\{n\;\text{if
$n$ odd};\; n-1\;\text{if $n$ even}\}\) (similarly for $n_c^\downarrow$).
Superscripts indicate known results:
\((\mathsf a)\) Ohga--Ito--Kolchinsky bound (2023)~\cite{OhgaItoKolchinsky2023},
itself a strengthened version of a conjecture by Barato--Seifert
(2017)~\cite{BaratoSeifert2017}.
\((\mathsf b)\) Kellogg--Stephens bound (1978)~\cite{kellogg1978complex}.
\((\mathsf c)\) Dmitriev--Dynkin bound (1946)~\cite{DmitrievDynkin1946}.
\((\mathsf d)\) Conjectured by Uhl--Seifert (2019)~\cite{UhlSeifert2019}.
\((\mathsf e)\) Follows from Gershgorin's Circle Theorem.
\((\mathsf f)\) Can be derived from the Karpelevich
theorem~\cite{KarpelevichCharacteristicRoots1951}.
See text for details.
}
\label{tab:downstream-bounds}
\centering

\begin{tblr}{
 width=\textwidth,
 colspec={Q[wd=0.045\textwidth,c]|Q[wd=0.115\textwidth,c]|Q[wd=0.075\textwidth,%
c]|Q[wd=0.16\textwidth,c]||Q[wd=0.195\textwidth,l]|Q[wd=0.105\textwidth,l]||%
Q[wd=0.18\textwidth,l]|Q[wd=0.084\textwidth,l]},
 cells={m},
 column{5}={rightsep=0.2pt},
 column{6}={leftsep=0.2pt},
 rowsep=0pt,
 belowsep=0pt,
 abovesep=2pt,
 colsep=1pt,
 rulesep=0.5pt,
}
\hline\hline
\SetCell[r=2]{c}\textbf{\(\mixc\)}
&
\SetCell[r=2]{c}\textbf{Constraint}
&
\SetCell[r=2]{c}\textbf{Region}
&
\SetCell[r=2]{c}\textbf{Quantity}
&
\SetCell[c=2]{c}\textbf{Thermodynamic bounds}
&
&
\SetCell[c=2]{c}\textbf{Non-thermodynamic bounds}
&
\\
&
&
&
&
\SetCell{c}\emph{Topology}
&
\SetCell{c}\emph{Size}
&
\SetCell{c}\emph{Topology}
&
\SetCell{c}\emph{Size}
\\
\hline
\(\displaystyle 1\)
&
\emph{Slow-decay oscillations}
&
\slowdecayicon
&
\(\displaystyle \frac{|\imL|}{\minusReL}\)
&
\(\displaystyle {}\boundle{a} \max_{c\in\CCC}\tAffc\cot\!\frac{\pi}{n_c}\)
&
\(\displaystyle {}\le\tAffcMax\cot\!\frac{\pi}{n}\)
&
\(\displaystyle {}\boundle{b} \max_{c\in\CCC}\cot\!\frac{\pi}{n_c}\)
&
\(\displaystyle {}\boundle{c} \cot\!\frac{\pi}{n}\)
\\
\hline
\(\displaystyle \frac{1}{2}\)
&
\emph{Uhl--Seifert\\ellipse}
&
\ellipseicon
&\(\displaystyle \frac{|\imL|}{\sqrt{\!-\!\reL(2\maxrate\!+\!\reL)}}\)
&
\SetCell[c=2]{c}\(\displaystyle {}\boundle{d} \tAffcMax\)
&
&
\SetCell[c=2]{c}\(\displaystyle {}\boundle{e} 1\)
&
\\
\hline
\!\!\(\displaystyle \frac{\!-\!\reL}{2\maxrate}\)
&
\emph{High-frequency oscillations}
&
\heighticon
&
\(\displaystyle \frac{|\imL|}{\maxrate}\)
&
\(\displaystyle {}\le
\max_{c\in\CCC}\tAffc\sin\!\frac{2\pi\lceil\tfrac{n_c-1}{4}\rceil}{n_c}\)
&
\(\displaystyle {}\le\tAffcMax\)
&
\(\displaystyle {}\le
\max_{c\in\CCC}\sin\!\frac{2\pi\lceil\tfrac{n_c-1}{4}\rceil}{n_c}\)
&
\(\displaystyle {}\boundle{e} 1\)
\\
\hline
\(\displaystyle 0\)
&
\emph{Fast-decay oscillations}
&
\fastdecayicon
&
\(\displaystyle \frac{|\imL|}{\dualReL}\)
&
\(\displaystyle {}\le \max_{c\in\CCC}\tAffc\tan\!\frac{\pi(\ncdown-1)}{2n_c}\)
&
\(\displaystyle {}\le\mathfrak{F}^* \cot\!\frac{\pi}{2n^\downarrow}\)
&
\(\displaystyle {}\le \max_{c\in\CCC}\tan\!\frac{\pi(\ncdown-1)}{2n_c}\)
&
\(\displaystyle {} \boundle{f} \cot\!\frac{\pi}{2n^\downarrow}\)
\\
\hline\hline
\end{tblr}
\end{table*}

We now present our main result for multicyclic systems. It gives a family of
eigenvalue bounds indexed by a parameter $\mixc\in[0,1]$ that interpolates
between constraints on slow-decay modes ($\mixc=1$) and fast-decay modes
($\mixc=0$). For each $\mixc$, the eigenvalue is constrained by the winding
number and affinity of some admissible nonequilibrium cycle, which we term the
\emph{certifying cycle}. The proof is in SM\ref*{app:proof-main}.

\newcommand\mainTheoremStmt{For every nonreal eigenpair $(\lambda,\bm u)$ and
parameter $\mixc \in [0,1]$, there is an admissible cycle $c$ with winding
number
$\winding_c\in \{1,\dots,\lceil n_c/2\rceil - 1\}$ such that
\begin{align*}
\mixc\frac{\absImL}{\minusReL}\tan\!\frac{\pi \winding_c}{n_c}
+
(1-\mixc)\frac{\absImL }{\dualReL}\cot\!\frac{\pi \winding_c}{n_c}
\le \tanh\!\frac{\Aff_c}{2n_c}.
\end{align*}}

\begin{theorem}
\label{thm:multicyclic-main}
\mainTheoremStmt{}
\end{theorem}

The certifying cycle has positive affinity ($\Aff_c>0$), since the left side of
the bound in Theorem~\ref{thm:multicyclic-main} is positive.
In the \emph{End Matter}, we outline an algorithm for identifying a certifying
cycle for any nonreal eigenpair and $\mixc\in[0,1]$. The structure of this
algorithm parallels the proof of Theorem~\ref{thm:multicyclic-main}, providing a
constructive perspective on the proof.

\section{Tradeoffs.---}
Theorem~\ref{thm:multicyclic-main} allows for a unified derivation of many
downstream eigenvalue bounds,  some of which recover known results and
conjectures from the literature.  Although generally weaker than
Theorem~\ref{thm:multicyclic-main}, these downstream bounds are simpler because
they do not require  detailed information about certifying cycles and winding
numbers.

Table~\ref{tab:downstream-bounds} summarizes several of these downstream bounds.
To derive these, we specialize Theorem~\ref{thm:multicyclic-main} to different
values of $\mixc$,  corresponding to different types of localization regions.
For each region, we derive a hierarchy of bounds that refer to different types
of information, and that capture different tradeoffs between thermodynamic,
spectral, and topological properties.
For instance, ``{Thermodynamic bounds}'' involve cycle affinities $\Aff_c$, while
``Topology''-dependent ones involve the number and sizes of cycles.
The strongest bounds, which use both thermodynamic and topological information,
are derived by applying Theorem~\ref{thm:multicyclic-main} at a given $\mixc$,
rearranging, and then maximizing over possible cycles $\CCC$ and winding numbers
\(\winding_c\in\{1,\dots,\lceil n_c/2\rceil-1\}\).
We derive weaker ``{Size}''-dependent bounds using the inequality \(n_c\le n\).
``{Non-thermodynamic bounds}'' are derived from thermodynamic ones via the
inequality \(\tanh(\Aff_c/2n_c)\le1\), which is tight for irreversible cycles;
such bounds are valid even for irreversible systems.

The derivation for \(\mixc=1/2\) involves an extra step. First, we use
Theorem~\ref{thm:multicyclic-main} to write
\begin{align*}
\frac12\Big(
\frac{|\imL|}{\minusReL}\tan\frac{\pi \winding_c}{n_c}
+
\frac{|\imL|}{\dualReL}\cot\frac{\pi \winding_c}{n_c}
\Big)
&\le
\tanh\frac{\Aff_c}{2n_c}\,.
\end{align*}
We then lower-bound the left side using the AM--GM inequality. Maximizing over
all cycles $\CCC$ gives
\begin{align}
\frac{\absImL}
{\sqrt{\minusReL\,(\dualReL)}}\le \max_c\tanh\frac{\Aff_c}{2n_c}\,,
\label{eq:us2}
\end{align}
which appears in the second row of Table~\ref{tab:downstream-bounds}.
Equation~\eqref{eq:us2} is the multicyclic version of the Uhl--Seifert
conjecture~\cite{UhlSeifert2019}, which has remained unresolved until now. (A
weaker version based on maximum edge force was recently shown in
Ref.~\cite{XuKolchinskyDelvenneIto2025}.)

\newcommand{\lUnif}{\lambda^{\mathrm{unif}}}
\newcommand{\lUnifConj}{\lambda^{\mathrm{unif}*}}

\section{Localization regions.---}
In our final results, we use Theorem~\ref{thm:multicyclic-main} to derive new
kinds of eigenvalue localization regions.
These regions reflect the interplay of constraints arising from
different cycles and interpolation parameter values $\mixc\in[0,1]$. The
detailed derivations are found in SM\ref*{app:convex-hull}.

We first define some notation.
We use $\CCCpos=\{c\in \CCC : \Aff_c > 0 \}$ to indicate the set of cycles with
positive affinity and $\mathcal W
:=
\prod_{c\in\CCCpos}
\big\{1,\dots,\lceil{n_c}/2\rceil -1 \big\}\subset
\mathbb{N}^{\vert\CCCpos\vert}$ to indicate the set of \emph{winding
configurations}, i.e., assignments of possible winding numbers to each cycle in
$\CCCpos$.  Finally, we use  \[\lUnif_{\Aff,n,\windA} := \maxrate[\cos({2\pi
\windA}/{n})-1+\ii\tanh({\Aff}/{2n})\sin({2\pi \windA}/{n})]\]
to indicate the eigenvalue of a uniform $n$-cycle with winding number $\windA$,
escape rate $\maxrate$, and cycle affinity $\Aff$.

We introduce two types of localization regions.
The first region depends on the cycle affinities and topological structure of
the generator (number and length of cycles).
To derive it, we show in the SM that  Theorem~\ref{thm:multicyclic-main} implies
that every eigenvalue,
regarded as a point in the complex plane, belongs to the convex hull
\[\Omega(\bm\windA):=\Conv\Big(\{-2\maxrate,0\}\cup\bigcup_{c\in\CCCpos}%
\{\lUnif_{\Aff_c,n_c,\windA_c},\lUnifConj_{\Aff_c,n_c,\windA_c}\}\Big)\]
for some winding configuration $\bm \windA\in\mathcal W$.
Taking the union of these convex hulls over possible winding configurations
gives an overall localization region.

\begin{theorem}
\label{thm:multicyclic-hull}
Every eigenvalue $\lambda$ belongs to $\bigcup_{\bm \windA\in\mathcal
W}\Omega(\bm\windA)$.
\end{theorem}

The resulting localization region is a (generally non-convex) polygon.
In the unicyclic case, $\CCCpos$ contains a single cycle and we recover the
polygon shown in Fig.~\ref{fig:unicyclic}(b).
In Fig.~\ref{fig:multicyclic-hull-small}(a), we illustrate
Theorem~\ref{thm:multicyclic-hull}  on generators that have three cycles.

\begin{figure}[t]
\centering
\includegraphics[width=\columnwidth]{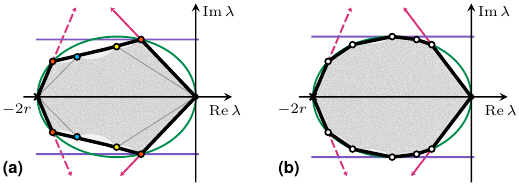}
\caption{(a) Theorem~\ref{thm:multicyclic-hull} illustrated for generators with
three cycles, $(n_c,\Aff_c)=(3,4),(4,6),(5,10)$.
Eigenvalues are localized to a topology- and affinity-dependent region (thick
black), defined as a union of convex hulls (thin gray lines) induced by the
three cycles (blue, yellow, and orange dots). We show four Topology-dependent
bounds from Table~\ref{tab:downstream-bounds}: $\tau=1$ (solid magenta),
$\tau=0$ (dashed magenta), $\tau=\minusReL/2\maxrate$ (purple lines), and
$\tau=1/2$ (Uhl--Seifert ellipse, green).
Gray regions show numerical eigenvalues from sampled generators with $n\le 10$
states and prescribed cycle data~\cite{sampling-methods}.
(b) Theorem~\ref{thm:topology-free-region} specifies a topology-free
localization region for size $n=6$ and $\mathfrak
F^*:=\max_{c}\tanh(\Aff_c/2n_c)=\tanh 1$.
It is the convex hull of points induced by  Farey fractions (white dots). We
show four Size-dependent bounds from Table~\ref{tab:downstream-bounds}, same
colors as in (a).
Gray regions show numerical eigenvalues from sampled generators with
$n=6,\mathfrak F^*\le\tanh 1$~\cite{sampling-methods}.}
\label{fig:multicyclic-hull-small}
\end{figure}

Our second localization result does not refer to the topology or affinities of
individual cycles, only to the system size $n$ and the largest normalized
affinity \(\mathfrak F^*:=\max_{c\in\CCC}\tanh(\Aff_c/2n_c)\) ($\mathfrak F^*=0$
if no cycles exist). It is expressed in terms of the \emph{Farey sequence}
$F_n$, i.e., the set of reduced fractions in $[0,1]$ with denominator at most
$n$ (e.g., \(F_4=\{0,1/4,1/3,1/2,2/3,3/4,1\}\)).

\begin{theorem}
\label{thm:topology-free-region}
Every eigenvalue \(\lambda\)
belongs to
\begin{align}
\Conv\Big\{
\maxrate\big[
\cos(2\pi q)-1+
\ii\,\mathfrak F^*\sin(2\pi q)
\big]
: q\in F_n
\Big\}.
\label{eq:weaker}
\end{align}
\end{theorem}

To derive this result, we use the fact that the convex hulls in
Theorem~\ref{thm:multicyclic-hull} are enlarged when $\tanh(\Aff_c/2n_c)$ is
replaced by $\mathfrak F^*=\max_{c}\tanh(\Aff_c/2n_c)$.
We remove dependence on topology by taking the convex hull of the resulting
points over all possible cycle sizes $n_c\in\{3,\dots,n\}$ and winding numbers.

We illustrate Theorem~\ref{thm:topology-free-region} using a schematic diagram
and numerical results in Fig.~\ref{fig:multicyclic-hull-small}(b).

Theorem~\ref{thm:multicyclic-hull} remains valid in the irreversible limit
$\Aff_c\to \infty$, where it gives a non-thermodynamic topology-dependent
localization region. Theorem~\ref{thm:topology-free-region} is also valid in
this regime, though in that case it is weaker than the region specified by the
Karpelevich theorem~\cite{KarpelevichCharacteristicRoots1951}.

\section{Discussion.---}
In this Letter, we showed that each complex eigenvalue of a Markov generator is
confined to a region determined by the cycle affinity and the winding number of
some  nonequilibrium cycle.
Moreover, in unicyclic systems, the winding number is determined by the ordered
eigenvalue index, leading to strong constraints on the fastest and slowest
eigenmodes.
Overall, our results elucidate how topological and thermodynamic properties
jointly constrain the eigenvalues of Markov generators.

Topological invariants such as winding numbers have appeared in other studies of
classical stochastic systems, in particular in work on topological phases and
protected edge currents in Markovian networks with repeating unit
cells~\cite{agudo2025topological,murugan2017topologically}. Unlike these earlier
definitions, our winding number is not a bulk invariant in real or parameter
space, but an eigenvector-dependent invariant defined on the cycles of a finite
Markov graph. Understanding the relationship between these different invariants
is an interesting direction for future work.

Our numerical experiments suggest that our localization bounds are not as sharp
as possible even in unicyclic systems,  as evidenced, for instance, by the empty
region in Fig.~\ref{fig:unicyclic}(a). Future work may investigate whether
tighter winding-resolved regions can be obtained. Other natural directions
include multicyclic generalizations of the winding-index correspondence
(Theorem~\ref{thm:ordered-winding}) and extensions of the one-dimensional
diffusion inequality~\eqref{eq:unicyclic-no-n-bound} to higher-dimensional
continuous systems.

\vspace{10pt}

\begin{acknowledgments}
We thank Jean-Charles Delvenne, Guo-Hua Xu, Ruicheng Bao, Hadrien Vroylandt, and
Chris Waggoner for helpful discussions. Several proofs were derived with
assistance from ChatGPT~Pro.
A.K. is partly supported by the John Templeton Foundation (grant 62828) and by
the European Union's Horizon 2020 research and innovation programme under the
Marie Sk{\l}odowska-Curie Grant Agreement No.~101068029.
N.O. is supported by JSPS KAKENHI Grants No.~23KJ0732 and No.~26KJ0175.
S.I. is supported by JSPS KAKENHI Grants No.~22H01141, No.~23H00467, and
No.~24H00834, JST ERATO Grant No.~JPMJER2302 and UTEC-UTokyo FSI Research Grant
Program.
\end{acknowledgments}

\vspace{10pt}
\begin{center}
\textbf{End Matter}
\end{center}

\section{Example.---}
\label{app:kai-response}
We briefly illustrate our results using a simple model of the
cyanobacterial biochemical clock~\cite{NguyenSeifertBarato2018}. This model considers
a single
KaiC hexamer with $i=0,1,\ldots,6$ phosphorylated sites and two global
conformations: an ``active'' conformation that favors KaiA-assisted
phosphorylation, and an ``inactive'' conformation that favors dephosphorylation.
For simplicity, and unlike Ref.~\cite{NguyenSeifertBarato2018}, we restrict
conformational changes to the fully dephosphorylated and fully
phosphorylated states. The model has 14 states in a unicyclic ring, see
Fig.~\ref{fig:kai}(a), representing the daily phosphorylation cycle (although
time in the model is measured in arbitrary units). We note that this idealized
model omits interactions between
hexamers.

The model depends on two thermodynamic quantities (in units
of $k_{\rm B}T$): the free energy released by ATP hydrolysis $\Delta\mu$, and
the conformational energy $E$, which favors the
active conformation at low phosphorylation and the inactive conformation
at high phosphorylation.
We also define a factor arising from rapid binding/unbinding kinetics of KaiA to
KaiC, $p_i=K_0a^i/(K_0a^i+A)$ for $i=0,\ldots,5$ and $p_6 = 1$, where $A$ is the free
KaiA concentration, $K_0$ is the KaiA dissociation constant for unphosphorylated
KaiC, and $a>1$ quantifies how binding weakens with each additional phosphate.
Finally, $\gamma$ and $\widetilde\gamma$ are kinetic prefactors.

We specify the kinetics of each reversible transition by its
forward rate $\fwd$ and force
$f=\ln(\fwd/\bwd)$, where $\bwd$ is the rate of the reverse transition.
We set $\fwd=\gamma(1-p_i),f=\Delta\mu/2-E/6+\ln(p_i/p_{i+1})$ for active
transitions $i\to i+1$;
$\fwd=\widetilde\gamma,f=\Delta\mu/2-E/6$ for inactive transitions $i\to i-1$;
$\fwd=1,f=E-\ln p_0$ for inactive-to-active switching at $i=0$; and $\fwd=1,f=E+\ln
p_6=E$ for active-to-inactive switching at $i=6$.
Following Ref.~\cite{NguyenSeifertBarato2018}, we take $E=20$, $\gamma=10$,
$\widetilde\gamma=1$, $a=10$,
$K_0=(3\times10^7)^{-1}$,
$A=0.04$ and $\Delta\mu=12$. The cycle affinity is $\Aff=72$.

\begin{figure}[!t]
\centering
\includegraphics[width=\columnwidth]{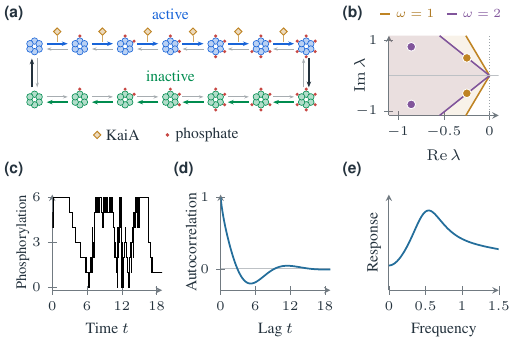}
\caption{\label{fig:kai}(a) Unicyclic model of the cyanobacterial biochemical
clock.
(b) Complex plane showing the slowest eigenvalue pair ($\omega=1$) and
second-slowest eigenvalue pair ($\omega=2$). Shaded regions represent
inequality~\eqref{eq:unicyclic-primal} from Theorem~\ref{thm:unicyclic-bounds}.
(c) A sample  trajectory of the phosphorylation level shows stochastic
oscillations.
(d) Steady-state autocorrelation function of the phosphorylation level.
(e) Linear-response gain of the average phosphorylation level under periodic
ATP/ADP driving, as a function of frequency.}
\end{figure}

The phosphorylation level $i$ is a natural observable in this system.
Figure~\ref{fig:kai}(c) plots a sample phosphorylation trajectory that exhibits
spontaneous stochastic oscillations.  The frequency and dephasing of these
oscillations can be quantified using the phosphorylation  autocorrelation
function $C(t)$, Fig.~\ref{fig:kai}(d). Importantly, the autocorrelation can be
decomposed into eigenmodes as $C(t)=\sum_{k=2}^{14} C_k(t)$, where $C_k(t)=c_k
e^{t \lambda_k }$ with an observable-dependent constant
$c_k$~\cite[Sec.~3.2]{hanggi1982stochastic}.
For any complex eigenmode, the coherence $\mathcal
Q_k=|\operatorname{Im}\lambda_k/\operatorname{Re}\lambda_k\vert $ measures the
``quality'' of the oscillations, since $\mathcal Q_k/(2\pi)$ is the number of
oscillations per decay time~\cite{BaratoSeifert2017}. $\mathcal Q_2$ quantifies
the quality of oscillations at the slowest timescale.

In our model, the two slowest eigenvalue pairs are
\begin{align}
\begin{aligned}
\lambda_{2,3}&=-0.248\pm0.497\ii & \mathcal Q_2&\approx 2.00 &\winding&=1\\
\lambda_{4,5}&=-0.856\pm0.810\ii &\mathcal Q_4&\approx 0.95 &\winding&=2\,,
\end{aligned}
\label{eq:second-eigenvalue}
\end{align}
where $\winding$ is the eigenvector winding number, see Fig.~\ref{fig:kai}(b).
Inequality~\eqref{eq:unicyclic-primal} from Theorem~\ref{thm:unicyclic-bounds}
gives $\mathcal Q_2\le 4.33, \mathcal Q_4\le 2.05$.
Both modes are within a factor of $2.2$ of their respective bounds.

In addition to spontaneous oscillations, the cyanobacterial clock can be
entrained by periodic driving of ATP/ADP concentrations~\cite{Rust2011}. We
model this by multiplying the active phosphorylation rates $i\to i+1$ by a weak
sinusoidal perturbation $e^{\epsilon \sin(\Omega t)}$ at frequency $\Omega$,
resulting in oscillations of the average phosphorylation level at the same
frequency.
Figure~\ref{fig:kai}(e) shows the linear-response gain $|\chi(\Omega)|$ of
output oscillations as a function of the frequency.

The response can be decomposed into eigenmodes as $\chi(\Omega)=\sum_{k=2}^{14}
\chi_k(\Omega),\chi_k(\Omega):={d_k}/({\ii\Omega-\lambda_k})$ with
observable-dependent constants  $d_k$~\cite[Sec.~5.2]{hanggi1982stochastic}.
Each complex eigenmode $k$ amplifies oscillatory inputs;
its squared gain \(|\chi_k(\Omega)|^2\) has a Lorentzian profile with center
\(\Omega^*=\operatorname{Im}\lambda_k\), half-width at half maximum
\(-\operatorname{Re}\lambda_k\), and
maximum squared relative gain $\vert \chi_k(\Omega^*)/\chi_k(0)\vert^2= 1+
\mathcal{Q}_k^2$.  Our eigenvalue localization bounds set constraints on these
characteristics of the resonant response of each eigenmode.

\newcommand\uniWadj{\widetilde{\W}}
\newcommand\uadj{\widetilde{\bm u}}
\newcommand\uadjxx{\widetilde{ u}}
\newcommand\madj{\widetilde{\winding}}

\newcommand\PhaseFunc{K}
\section{Derivation of Theorem~\ref{thm:unicyclic-bounds}.---}
The inequality part of Theorem~\ref{thm:unicyclic-bounds}
follows as a simple corollary of Theorem~\ref{thm:multicyclic-main}, with
inequalities~\eqref{eq:unicyclic-primal} and~\eqref{eq:unicyclic-dual} given by
taking $\mixc=1$ and $\mixc=0$. In the unicyclic case,
Lemma~\ref*{lem:upper-half-plane} in SM\ref*{app:unicyclic-proofs} shows that
the eigenvector ratios $u_\yy/u_\xx$ lie in a half-plane of the complex plane; 
thus, the cycle is admissible. Also, the phase vector $e^{\ii \branchParam}$ of the branch parameter $\branchParam$~\eqref{eq:winding-multicycle} lies in this half-plane, so 
the multicyclic winding number~\eqref{eq:winding-multicycle} agrees with the
unicyclic one~\eqref{eq:winding-unicycle}.

However, to build intuition, we provide a self-contained derivation of the
unicyclic result that does not reference Theorem~\ref{thm:multicyclic-main}. For
simplicity, we consider the case where $\imL>0$, $0<\Aff<\infty$, and $\bm u$ is
a left eigenvector. Theorem~\ref{thm:ordered-winding}, proved in
SM\ref*{app:unicyclic-proofs}, shows that the right eigenvector case is
equivalent.
The irreversible case $\Aff=\infty$ is treated in
SM\ref*{app:unicyclic-equality-infinite}. The DB case $\Aff=0$ is trivial (all
eigenvalues are real).

We first derive Eq.~\eqref{eq:unicyclic-primal}. \begingroup
Lemma~\ref*{lem:upper-half-plane} in SM\ref*{app:unicyclic-proofs} shows that
$u_\xx \ne 0$ and $\im(u_{\xx+1}/u_\xx) >0$ for all $\xx$. Thus, we may
define eigenvector ratios $q_\xx :=u_{\xx+1}/u_\xx$ having polar representation
$q_\xx = e^{-\xlog_\xx+\ii\phase_\xx}$ with $\phase_\xx \in(0,\pi)$. We may then
rearrange the eigenvector equation $\W^\top\bm u = \lambda\bm u$ as
\begin{align}
\lambda
=
\fwd_\xx (q_{\xx}-1) +\bwd_\xx (1/q_{\xx-1}-1)\,.
\label{eq:generator-m4}
\end{align}
\endgroup
Next, we write the eigenvalue in polar form as $\lambda=-|\lambda|e^{-\ii\psi}$
for $\psi\in(0,\pi/2)$ and define the quantities
$A_\xx:=\im[e^{\ii\psi}(q_\xx-1)]=e^{-\xlog_i}\sin(\phase_i+\psi)-\sin\psi$ and
$B_\xx:=\im[e^{\ii\psi}(1-1/q_\xx)]=e^{\xlog_i}\sin(\phase_i-\psi)+\sin\psi$.
These satisfy
\begin{align}
A_\xx=\frac{\bwd_\xx}{\fwd_\xx}B_{\xx-1},\quad
|q_\xx|^2B_\xx-A_\xx=\sin\psi\,|q_\xx-1|^2>0.
\label{eq:asdfz}
\end{align}
The first equality follows from multiplying Eq.~\eqref{eq:generator-m4} by
$e^{\ii\psi}/\fwd_\xx$, taking the imaginary part, and rearranging. The second
follows from the definitions of $A_\xx$ and $B_\xx$. Since $\prod_\xx|q_\xx|=1$,
some $|q_\yy|\ge1$. Together with $\im q_\yy>0$, this puts $1-1/q_\yy$ in the
first quadrant, so $B_\yy>0$. Propagating Eqs.~\eqref{eq:asdfz} then implies
that all $A_\xx>0,B_\xx>0$.

Taking logarithms, summing across the cycle, and rearranging expresses the cycle
affinity in terms of the eigenvalue and eigenvector entries,
$\Aff\equiv\sum_\xx\ln({\fwd_\xx}/{\bwd_\xx})
=\sum_\xx\ln({B_\xx}/{A_\xx})$.
Using the definitions of $A_\xx,B_\xx$,  we may rewrite the affinity as
$\Aff=\sum_\xx\PhaseFunc_\psi(\phase_\xx,\xlog_\xx)$, where
\begin{align}
\PhaseFunc_\psi(\phase,\xlog)
:=
\ln\frac{e^{\xlog}\sin(\phase-\psi)+\sin\psi}
{e^{-\xlog}\sin(\phase+\psi)-\sin\psi}\,.
\label{eq:Gdef}
\end{align}
Since \(0<\theta_\xx<\pi\) and \(A_\xx>0\), every \((\theta_\xx,\xi_\xx)\) lies
in the domain specified by \(0<\theta<\pi-\psi\) and
\(e^\xi<\sin(\theta+\psi)/\sin\psi\).
Corollary~\ref*{cor:K-convexity} in SM\ref*{app:convexity} shows that the
function \(\PhaseFunc_\psi\) is strictly convex over this convex domain.
Applying Jensen's inequality and using $\sum_\xx\phase_\xx=2\pi
\winding,\sum_\xx\xlog_\xx=0$ gives
\begin{align}
\Aff/n=\frac1n\sum_\xx\PhaseFunc_\psi(\phase_\xx,\xlog_\xx)
\ge\PhaseFunc_\psi(2\pi\winding/n,0).
\label{eq:zcvdd}
\end{align}
After rearranging, Eq.~\eqref{eq:zcvdd} is  equivalent to
Eq.~\eqref{eq:unicyclic-primal}.

Equation~\eqref{eq:unicyclic-dual} is derived by applying a similar technique
to a reflected construction. Let us define a reflected ``pseudo-eigenvalue''
$\lambda^\prime = -2\maxrate - \lambda^*=-|\lambda'|e^{-\ii\psi'}$, with
$\psi'\in(0,\pi/2)$.
\begingroup
Rearranging Eq.~\eqref{eq:generator-m4} gives
\begin{align*}
-\lambda'
=
\fwd_\xx(q_{\xx}^*+1)
+\bwd_\xx(1/q_{\xx-1}^*+1)
+2(\maxrate-\maxrate_\xx) ,
\end{align*}
\endgroup
where $\maxrate_\xx:=\fwd_\xx + \bwd_\xx$ is the local escape rate.
We multiply both sides by
$e^{\ii\psi'}$,
take the imaginary part, and use the inequality
$2(\maxrate-\maxrate_\xx)\sin\psi' \ge 0$ to write
\begin{align*}
\fwd_\xx \im\bigl(e^{\ii\psi'}( -q_\xx^*-1)\bigr)
\ge
\bwd_\xx \im\bigl(e^{\ii\psi'}(1+ 1/q_{\xx-1}^{*})\bigr)\,.
\end{align*}
The same propagation argument, applied to the reflected ratios $-q_\xx^*$, shows
that both sides are positive. Taking logarithms and summing across the cycle
gives $\Aff\ge\sum_\xx\PhaseFunc_{\psi'}(\pi-\phase_\xx,\xlog_\xx)$. Finally,
applying Jensen's inequality gives
$\Aff/n\ge\PhaseFunc_{\psi'}(\pi-2\pi\winding/n,0)$, which is equivalent to
Eq.~\eqref{eq:unicyclic-dual}.

It is easy to verify that a uniform cycle saturates both
bounds~\eqref{eq:unicyclic-primal} and~\eqref{eq:unicyclic-dual}. To prove the
converse, assume that either bound is saturated. This requires equality in the
respective Jensen inequality, applied to $\PhaseFunc_\psi$ or
$\PhaseFunc_{\psi'}$. Strict convexity and the sum constraints then force
$\phase_\xx=2\pi\winding/n$ and $\xlog_\xx=0$ for every $\xx$, so that all
$q_\xx=e^{2\pi\ii\winding/n}$. Substituting into Eq.~\eqref{eq:generator-m4} and
taking the real part implies that $\fwd_\xx+\bwd_\xx$ is uniform across the
cycle; taking the imaginary part implies that $\fwd_\xx-\bwd_\xx$ is uniform.
Thus,  the rates must be uniform.

\section{Certifying cycle in Theorem~\ref{thm:multicyclic-main}.---}
\label{app:certifying-algorithm}

The derivation of Theorem~\ref{thm:unicyclic-bounds} above expressed the
cycle affinity as a sum of convex terms, then applied Jensen's inequality.
The proof of the multicyclic Theorem~\ref{thm:multicyclic-main}
in SM\ref*{app:proof-main} uses a similar strategy. For each choice of
$\mixc\in[0,1]$, eigenvalue $\lambda$, and eigenvector $\bm u$, we identify an
admissible certifying cycle, then derive the eigenvalue bound using a convexity
argument.

The certifying cycle can be identified using a three-step algorithm. Here, we
consider a rate matrix $\W\in\RR^{n\times n}$ with maximum escape rate
$\maxrate$.
For simplicity, we assume that $\imL >0$ and that $\bm u$ is a left eigenvector.
Otherwise, we first pass to the conjugate eigenpair and/or the adjoint rate
matrix (for details, see SM\ref*{app:righteig}).

\newcommand\algoGraph{\mathcal G}
\newcommand\algoVert{\mathcal V}
\newcommand\algoEdge{\mathcal E}

First, define the matrix
\begin{align*}
M_{\yy\xx}
:=
\operatorname{Im}(u_\yy u_\xx^*)
+
(\alpha -\beta)\operatorname{Re}(u_\yy u_\xx^*)
-(\alpha +\beta)|u_\xx|^2\,,
\end{align*}
where we introduced $\alpha :=
\mixc\imL/(\minusReL)$ and $\beta:={(1-\mixc)\imL}/{(\dualReL)}$. Also, using
this matrix, define a weighted directed graph \(\algoGraph =(\algoVert,
\algoEdge ,\bm g)\) as
\begin{align*}
\algoVert &:=\{\xx :u_\xx\ne0 \}\\
\algoEdge &:=\{\yxedge:\xx,\yy\in \algoVert,\ \W_{\yy\xx}>0, u_\yy\ne \pm u_\xx,
M_{\yy\xx}\ge0\}\\
\oyxedge&:=
\begin{cases}
\ln\dfrac{\W_{\yy\xx}M_{\yy\xx}}{-\W_{\xx\yy}M_{\xx\yy}}
& M_{\yy\xx}>0,\W_{\xx\yy}>0\\
-\infty
& \text{otherwise}
\end{cases} \enspace \forall\yxedge \in \algoEdge.
\end{align*}

Second, find a cycle $c$ in $\algoGraph$ with maximum mean weight $\sum_{\yxedge
\in c} \oyxedge/n_c$ using a standard graph-theoretic
algorithm~\cite{dasdan1998faster}.
If the maximum-mean-weight cycle $c$ has nonnegative weight $\sum_{\yxedge \in
c} \oyxedge\ge 0$, take it as the certifying cycle. This cycle will only have
reversible edges, since  $\algoGraph$ assigns weight $-\infty$ to irreversible
edges.

Third, if the previous step did not return a cycle $c$ with nonnegative weight,
then the certifying cycle must contain an irreversible edge.
To find this cycle, identify the strongly connected components of $\algoGraph$,
take any irreversible edge $\yxedge$ whose endpoints $\xx,\yy$ lie in the same
component, and concatenate $\yxedge$ with any simple directed path from $\yy$
back to $\xx$ within that component.

The proof of Theorem~\ref{thm:multicyclic-main} in SM\ref*{app:proof-main} uses
a graph-theoretic argument to show that the above procedure always returns an
admissible certifying  cycle. Once this cycle is found, its winding number can
be computed using
Eq.~\eqref{eq:winding-multicycle}.

\end{document}


\setcounter{theorem}{5}
\appendix
\setcounter{page}{1}
\renewcommand{\thepage}{S\arabic{page}}
\renewcommand{\appendixname}{SM}
\onecolumngrid
\begin{center}
{\large\bfseries Supplemental Material for\par}
{\large\bfseries Cycle affinity and winding localize eigenvalues of Markov
generators\par}
\vspace{3pt}
Artemy Kolchinsky, Naruo Ohga, and Sosuke Ito\par
\end{center}
\vspace{10pt}
\twocolumngrid
\let\section\standardsection
\let\subsection\standardsubsection
\setcounter{secnumdepth}{2}
\setcounter{section}{0}
\setcounter{equation}{0}
\renewcommand{\thesection}{\arabic{section}}
\renewcommand{\thesubsection}{\thesection.\arabic{subsection}}
\renewcommand{\theequation}{S\arabic{equation}}
\providecommand{\theHequation}{\theequation}
\renewcommand{\theHequation}{S\arabic{equation}}
\makeatletter
\@removefromreset{equation}{section}
\renewcommand{\p@subsection}{}
\makeatother

\section{UNICYCLIC RESULTS: DISCRETE}
\label{app:unicyclic}

\subsection{Proof of Theorem~\ref{thm:ordered-winding}}
\label{app:unicyclic-proofs}

Here we prove Theorem~\ref{thm:ordered-winding}, which states that the unicyclic
winding number is determined by the sorted eigenvalue index. We also prove that
the unicyclic winding number can be equivalently defined in terms of either the
left or right eigenvector.

We begin with a helpful lemma. It refers to the cycle
affinity,
\[
\Aff:=\ln\frac{\prod_\xx{\fwd_\xx}}{\prod_\xx{\bwd_\xx}}\in [-\infty,\infty]\,,
\]
which is well-defined as long as either the numerator or denominator is strictly
positive.
We use the convention that $\sign \pm\infty=\pm1$.

\begin{lemma}
\label{lem:upper-half-plane}
Given $\fwdVec\in\RR_{\ge 0}^n, \bwdVec\in\RR_{\ge 0}^n,\bm
u\in\mathbb{C}^n\setminus\{\zz\}$ and nonreal $\lambda \in\mathbb{C}$, suppose
that
for all $\xx=1,\dots,n$,
\begin{align}
\lambda u_\xx
=
\fwd_\xx u_{\xx+1}+\bwd_\xx u_{\xx-1}-(\fwd_\xx+\bwd_\xx)u_\xx \,.
\label{eq:baas-m4}
\end{align}
Then, $\Aff$ is well-defined and nonzero. Also, $u_\xx \ne 0$ and
$\sign[\im({u_{\xx+1}}/{u_\xx})]=\sign(\Aff \imL)$ for all $\xx$.
\end{lemma}

\begin{proof}
For convenience, we define $\gamma_\xx:=\imL\,\im(u_{\xx+1}u_\xx^*)$.
Multiplying both sides of
Eq.~\eqref{eq:baas-m4} by $(\imL)\,u_\xx^*$ and taking imaginary parts gives
\begin{align}
\fwd_\xx\gamma_\xx-\bwd_\xx\gamma_{\xx-1}=(\imL)^2|u_\xx|^2\,.
\label{eq:dsaf2}
\end{align}
At least one $\gamma_\xx\ne 0$, since otherwise
Eq.~\eqref{eq:dsaf2} would imply $\bm u=\zz$. Note that $\gamma_\xx\ne 0$
implies
$u_\xx\ne 0$ and $u_{\xx+1}\ne 0$.

We now show that all $\gamma_\xx \ne 0$ and share the sign of $\Aff$.
Suppose first that some $\gamma_\yy>0$. Equation~\eqref{eq:dsaf2} then gives
$\fwd_{\yy+1}\gamma_{\yy+1}=\bwd_{\yy+1}\gamma_\yy+(\imL)^2|u_{\yy+1}|^2>0$, so
$\fwd_{\yy+1}>0$ and
$\gamma_{\yy+1}>0$. Propagating forward around the cycle implies that
$\gamma_\xx>0$ and $\gamma_\xx>(\bwd_\xx/\fwd_\xx)\gamma_{\xx-1}$ for every
$\xx$. Iterating around the
cycle gives $\gamma_\xx>e^{-\Aff}\gamma_\xx$, hence $\Aff>0$.
Suppose instead that some $\gamma_\yy<0$. Then Eq.~\eqref{eq:dsaf2} gives
$\bwd_\yy\gamma_{\yy-1}=\fwd_\yy\gamma_\yy-(\imL)^2|u_\yy|^2<0$, so $\bwd_\yy>0$
and
$\gamma_{\yy-1}<0$. Propagating backward around the cycle implies that
$\gamma_\xx<0$,
$\bwd_\xx>0$, and $|\gamma_{\xx-1}|>(\fwd_\xx/\bwd_\xx)|\gamma_\xx|$ for every
$\xx$.
Iterating around the cycle gives $|\gamma_\xx|>e^{\Aff}|\gamma_\xx|$, hence
$\Aff<0$.

Finally, $\gamma_\xx \ne 0$ for all $\xx$ implies that $u_\xx \ne 0$ for all
$\xx$. Also,
$\sign[\im(u_{\xx+1}/u_\xx)]=\sign(\gamma_\xx/\imL)=\sign(\Aff\,\imL)$.
\end{proof}

For a left eigenvector $\bm u$ of an irreducible unicyclic generator with
nonreal eigenvalue, Lemma~\ref{lem:upper-half-plane} gives $u_\xx\ne 0$ and
places all $\Arg(u_{\xx+1}/u_\xx)$ in either $(0,\pi)$ or $(-\pi,0)$. Also,
$m:=\sum_{\xx=1}^n\Arg({u_{\xx+1}}/{u_\xx})/2\pi$
obeys  $e^{\ii 2\pi m}=\prod_\xx u_{\xx+1}/u_\xx=1$,  so $m$ is an integer.
Consequently, $\windFunc(\bm u)=|m|_n=|m|$ is well-defined and lies in
$\{1,\dots,\lceil n/2\rceil-1\}$.

\begin{reptheoremnum}{thm:ordered-winding}
\orderedWindingStmt{}
\end{reptheoremnum}

\begin{proof}
We introduce some relevant definitions. A matrix is called \emph{Metzler} if all
of its off-diagonal entries are nonnegative. It is called \emph{strictly sign
regular} of order $r$, written SSR$_r$, if all of its minors of order $r$ are
nonzero
and have the same sign~\cite[Def.~1]{BenAvrahamSharonZaraiMargaliot2020}. Also,
for any vector $\bm {x}\in\RR^n$ with no zero entries, the \emph{sign variation}
$s(\bm x)$ is the
number of sign changes in the linear sequence $(x_1,\dots,x_n)$. The
\emph{cyclic sign variation} $s_c(\bm x)$ is the number of sign
changes in the cyclic sequence $(x_1,\dots,x_n,x_1)$. The cyclic sign variation
$s_c(\bm x)$ is even and equal to either $s(\bm x)$ or $s(\bm x)+1$.

An irreducible unicyclic rate matrix $\W \in \RR^{n\times n}$ is Metzler and can
be ordered so that it is tridiagonal with corners. Then, Thm.~5 from
Ref.~\cite{BenAvrahamSharonZaraiMargaliot2020} shows that the transition matrix
$A(t):=e^{t\W}$ for any $t>0$ obeys the \emph{strong cyclic variation
diminishing property} (SCVDP). For the definition of SCVDP, see~\cite[Def.~2 and
Eq.~(14)]{BenAvrahamSharonZaraiMargaliot2020}. Because $A(t)$ has the SCVDP,
Thm.~2 from Ref.~\cite{BenAvrahamSharonZaraiMargaliot2020} implies that it is
SSR$_r$ for all odd $r\in\{1,\dots,n\}$.

Let us write the eigenvalues of $A(t)$, counted with algebraic multiplicity, in
terms of the eigenvalues of $\W$ as $e^{t\lambda_k}$.
Since $A(t)$ is SSR$_r$ for all odd $r$, Thm.~3(ii) from
Ref.~\cite{AlseidiMargaliotGarloff2019} implies its eigenvalues can be ordered
into strictly sorted pairs,
\[
1=\vert e^{t\lambda_1} \vert > \vert e^{t\lambda_2} \vert \ge \vert
e^{t\lambda_3} \vert > \vert e^{t\lambda_4} \vert \ge \vert e^{t\lambda_5}
\vert> \dots\,.
\]
Since $\vert e^{t\lambda_k} \vert=e^{t\re\lambda_k}$, this is equivalent to
\begin{align}
0=\lambda_1 > \re\lambda_2 \ge \re\lambda_3 > \re \lambda_4\ge \re \lambda_5 >
\cdots\,.
\label{eq:ordegapp}
\end{align}

Consider any nonreal eigenvalue $\lambda_k$. Its conjugate eigenvalue
$\lambda_k^*$ has the same real part, so Eq.~\eqref{eq:ordegapp} implies that
$(\lambda_k,\lambda_k^*)$ belong to block
$\{\lambda_{2\evblockix},\lambda_{2\evblockix+1}\}$ with $\evblockix=\lfloor
k/2\rfloor$. In particular, $\lambda_k$ is algebraically simple, since a
repeated nonreal eigenvalue and its conjugate would occupy at least four
positions with the same real part.

Let $\bm u$ be a left eigenvector of $\W$ associated with $\lambda_k$, and note
that it is also a left eigenvector of the matrix $A(t)$ associated with the
$k$-th eigenvalue for $k\in \{2\evblockix,2\evblockix+1\}$. Since $A(t)$ is
SSR$_r$ for all odd $r$, Thm.~3(v) from Ref.~\cite{AlseidiMargaliotGarloff2019}
implies that every
$\bm x\in \operatorname{span}_{\RR}\{\re \bm u,\im \bm u\}$ without
zero entries has $s(\bm x)\in \{2\evblockix-1,2\evblockix\}$. Since $s_c(\bm x)$
is even and equal either to $s(\bm x)$ or to $s(\bm x)+1$, we have
$s_c(\bm x)=2\evblockix=2\lfloor k/2\rfloor$.

Next, we take $\bm \evblockfunc = \re \bm u$ and assume that $\evblockfunc_\xx
\ne 0$ for all $\xx$. This assumption is made without loss of generality, since
we may choose the eigenvector's phase $\bm u \mapsto e^{\ii \eta} \bm u$ to
satisfy it (since $u_\xx\neq 0$ for all $\xx$ by
Lemma~\ref{lem:upper-half-plane}, only finitely many values of $\eta\in
[-\pi,\pi]$ are excluded).
We also introduce the lifted phase
\[
\Theta_1=\Arg u_1 \qquad \Theta_{\xx+1} = \Theta_{\xx} +
\Arg\frac{u_{\xx+1}}{u_\xx} \,.
\]
By the observation following Lemma~\ref{lem:upper-half-plane}, the total
increment is $\Theta_{n+1}-\Theta_1=2\pi m$ for a nonzero integer $m$ with
$|m|<n/2$.

The cyclic sign variation $s_c(\bm \evblockfunc)$ is the number of times that
$\evblockfunc_{\xx+1}$ and $\evblockfunc_{\xx}$ have opposite signs. Noting that
$\evblockfunc_\xx = |u_\xx|\cos \Theta_\xx$, it is also the number of times that
\(\cos\Theta_\xx\) and \(\cos\Theta_{\xx+1}\) have opposite signs. This implies
that $s_c(\bm \evblockfunc)$ is equal to the number of ``crossings'', where the
open interval with endpoints \(\Theta_\xx\) and \(\Theta_{\xx+1}\) contains a
point \(\pi/2+q\pi\) for some \(q\in\mathbb Z\).
From Lemma~\ref{lem:upper-half-plane}, all
$\im (u_{\xx+1}/u_\xx) \ne 0$ and have the same sign, therefore either all
increments $\Theta_{\xx+1}-\Theta_\xx$ lie in $ (0,\pi)$ or they all lie in
$(-\pi,0)$.
Therefore, \(\Theta_\xx\) changes monotonically and each interval has magnitude
less than $\pi$, so
each interval can contain at most one such crossing. Since the total increment
is an integer multiple of $2\pi$, the cyclic sign variation is
\begin{align*}
s_c(\bm \evblockfunc)=\frac{1}{\pi}\vert \Theta_{n+1}-\Theta_1\vert \,.
\end{align*}
At the same time, we have
\begin{align*}
\frac{1}{2\pi}\vert \Theta_{n+1}-\Theta_1\vert &= \Big\vert \frac{
\Theta_{n+1}-\Theta_1}{2\pi}\Big\vert_n\\
&=\Big\vert\frac{1}{2\pi}\sum_{\xx=1}^n\Arg\frac{u_{\xx+1}}{u_\xx}
\Big\vert_n=\winding\,,
\end{align*}
where we first used $\vert \Theta_{n+1}-\Theta_1\vert<n\pi$ and $\vert
m\vert_n=|m|$ for $|m|\le n/2$.
Combining $s_c(\bm \evblockfunc)=2
\winding$ and $s_c(\bm \evblockfunc)=2\lfloor k/2\rfloor$ gives $\windFunc(\bm
u)=\lfloor k/2\rfloor$.

Finally, let $\bm v$ be a right eigenvector of $\W$ associated with $\lambda_k$.
Since $\W$ is irreducible, the steady state $\ppss$ has full support. Define
$\Pi:=\operatorname{diag}(\ppss)$ and the adjoint rate matrix
$\uniWadj:=\Pi\W^\top\Pi^{-1}$. This is again an irreducible unicyclic
generator, similar to $\W^\top$ and hence with the same spectrum as $\W$.%

The vector $\uadj:=\Pi^{-1}\bm v$ is a left eigenvector of $\uniWadj$ associated
with $\lambda_k$. Applying the left-eigenvector result just established gives
$\windFunc(\uadj)=\lfloor k/2\rfloor$. Moreover,
Lemma~\ref{lem:upper-half-plane} gives $\uadjxx_\xx\ne 0$, and hence $v_\xx\ne
0$, for all $\xx$. Since the steady-state probabilities are positive,
\[
\Arg\frac{\uadjxx_{\xx+1}}{\uadjxx_\xx}
=\Arg\frac{v_{\xx+1}/\ppssx{\xx+1}}{v_\xx/\ppssx{\xx}}
=\Arg\frac{v_{\xx+1}}{v_\xx}\,.
\]
Thus $\windFunc(\bm v)=\windFunc(\uadj)=\lfloor k/2\rfloor=\windFunc(\bm u)$.

\end{proof}

As mentioned in the main text, Theorem~\ref{thm:ordered-winding} can be
understood from the perspective of topological invariants. The strict
inequalities~\eqref{eq:ordegapp} between consecutive blocks guarantee that no
eigenvalue can move between blocks under continuous deformation within the
unicyclic class. In addition, the block index determines the eigenvector's
cyclic sign variation, which in turn determines the winding number of any
complex eigenvalue. Together, these facts imply that the winding number is
topologically protected in unicyclic systems.

\subsection{Theorem~\ref{thm:unicyclic-bounds} with
  \texorpdfstring{$\Aff=\infty$}{F = infinity}}
\label{app:unicyclic-equality-infinite}

With a small modification of the derivation in the \emph{End Matter}, we extend
the inequalities and equality characterization in
Theorem~\ref{thm:unicyclic-bounds} to the $\Aff=\infty$ case.
At least one backward rate $\bwd_\xx$ then vanishes. As in the
\emph{End Matter}, let $\bm u$ be a left eigenvector
with eigenvalue $\lambda$ satisfying $\imL>0$. Write the ratio on successive
states as $q_\xx:=u_{\xx+1}/u_\xx=e^{-\xlog_\xx+\ii\phase_\xx}$.
Here $\phase_\xx\in(0,\pi)$ is the phase increment along edge $\xx$
(Lemma~\ref{lem:upper-half-plane}), and $\xlog_\xx=-\ln|q_\xx|$ is the decrease
in log magnitude. The ratios telescope around the ring, so $\prod_\xx q_\xx=1$.
Thus, $\sum_\xx\xlog_\xx=0$ and
$\vartheta:=n^{-1}\sum_\xx\phase_\xx=2\pi\winding/n$.

Let $\psi\in(0,\pi/2)$ be the angle between $\lambda$ and the negative
real axis, so $\tan\psi=\imL/(-\reL)$. The sign argument following
Eq.~\eqref{eq:asdfz} gives $A_\xx:=\im[e^{\ii\psi}(q_\xx-1)]\ge0$
also when a backward rate vanishes. Substituting the polar form of $q_\xx$
gives $e^{-\xlog_\xx}\sin(\phase_\xx+\psi)\ge\sin\psi>0$.
It follows that $0<\phase_\xx<\pi-\psi$ and
\begin{align}
\xlog_\xx\le f_\psi(\phase_\xx),\qquad
f_\psi(\phase):=\ln\frac{\sin(\phase+\psi)}{\sin\psi}.
\label{eq:unicyclic-infinite-pointwise}
\end{align}
For the mean phase $\vartheta$, we obtain
\begin{align}
0=\frac1n\sum_\xx\xlog_\xx
\le \frac1n\sum_\xx f_\psi(\phase_\xx)
\le f_\psi(\vartheta).
\label{eq:unicyclic-infinite-exponential}
\end{align}
Here the first inequality follows by averaging
Eq.~\eqref{eq:unicyclic-infinite-pointwise}, and the second follows from the
concavity of $f_\psi$.
Thus $\sin(\vartheta+\psi)\ge\sin\psi$. Since both angles lie in
$(0,\pi)$ and $\vartheta>0$, this requires $\vartheta+\psi\le\pi-\psi$.
Equivalently, $\tan\psi\le\cot(\vartheta/2)$, which is
Eq.~\eqref{eq:unicyclic-primal} at $\Aff=\infty$.

For Eq.~\eqref{eq:unicyclic-dual}, let $r$ be the maximum escape rate
and define the reflected value
$\lambda'=-2r-\lambda^*=-|\lambda'|e^{-\ii\psi'}$, with $\psi'\in(0,\pi/2)$.
The reflected ratio is
$-q_\xx^*=e^{-\xlog_\xx+\ii(\pi-\phase_\xx)}$.
The reflected sign argument in the \emph{End Matter} gives
$\im[e^{\ii\psi'}(-q_\xx^*-1)]\ge0$. Thus
Eq.~\eqref{eq:unicyclic-infinite-pointwise} holds with $\psi'$ and
$\pi-\phase_\xx$ in place of $\psi$ and $\phase_\xx$. The mean phase is now
$\pi-\vartheta$, so
averaging as above gives
\begin{align}
\tan\psi'\le\cot\frac{\pi-\vartheta}{2}
=\tan\frac{\vartheta}{2}.
\label{eq:unicyclic-infinite-reflected}
\end{align}
Since $\tan\psi'=\imL/(\dualReL)$, this is
Eq.~\eqref{eq:unicyclic-dual} at $\Aff=\infty$.

Finally, we consider equality conditions. The uniform cycle saturates both
bounds~\eqref{eq:unicyclic-primal} and~\eqref{eq:unicyclic-dual}, including at
infinite affinity. To consider the converse, assume
that~\eqref{eq:unicyclic-primal} is saturated at $\Aff=\infty$. Then
$\psi+\vartheta/2=\pi/2$, so $f_\psi(\vartheta)=0$ and both inequalities in
Eq.~\eqref{eq:unicyclic-infinite-exponential} hold with equality.
Strict concavity of $f_\psi$ forces every $\phase_\xx=\vartheta$.
Equation~\eqref{eq:unicyclic-infinite-pointwise} then gives $\xlog_\xx\le0$, and
their zero sum forces every $\xlog_\xx=0$.
Thus, $q_\xx=e^{\ii\vartheta}$. Substituting into Eq.~\eqref{eq:generator-m4}
and taking real and imaginary parts, as in the \emph{End Matter}, again forces
uniform forward and backward rates.

The same argument applies if Eq.~\eqref{eq:unicyclic-dual} is saturated at
$\Aff=\infty$. Applied to the reflected phases $\pi-\phase_\xx$, it again gives
every $\xlog_\xx=0$ and $\phase_\xx=\vartheta$.
Thus, $q_\xx=e^{\ii\vartheta}$, implying that forward and backward rates are
uniform.

\section{UNICYCLIC RESULTS: CONTINUOUS}
\label{app:fp}

\newcommand{\rhoss}{\rho_{\mathrm{ss}}}
\newcommand{\jss}{j_{\mathrm{ss}}}

Consider the Fokker--Planck equation on the unit circle,
\begin{align}
\partial_t\rho=\mathcal L\rho
:=-\partial_x[f(x)\rho(x,t)]+\partial_x^2[Q(x)\rho(x,t)],
\label{eq:app-fp-cont}
\end{align}
where $x \in [0,1]$ and all coefficients are periodic with period one. We assume
that $Q(x)>0$ everywhere and that the drift $f(x)$ and diffusion $Q(x)$ fields
are smooth. We define the continuous cycle affinity,
\begin{align}
\Aff:=\int_0^1\frac{f(x)}{Q(x)}\,dx,
\label{eq:app-cont-aff}
\end{align}
and we choose the orientation of the cycle so that $\Aff\ge0$.

In Theorem~\ref{thm:cont-bound-fp}, we show that the eigenvalues of this
Fokker--Planck operator obey inequality~\eqref{eq:unicyclic-no-n-bound}, once it
is stated using the continuous cycle affinity~\eqref{eq:app-cont-aff} and the
eigenfunction winding number, defined as
\begin{align}
\windFunc(u)
:=
\frac{1}{2\pi}\left|
\int_0^1
\im\frac{\partial_x u(x)}{u(x)}\,dx
\right| \,,
\label{eq:app-cont-winding-def}
\end{align}
which is well-defined if $u$ is nonvanishing everywhere. In
Theorem~\ref{thm:cont-index}, we prove the direct relationship between the
eigenfunction winding number and the eigenvalue index, analogous to
Theorem~\ref{thm:ordered-winding}.

We will refer to the adjoint generator, defined as
\begin{align}
\mathcal L^\dagger u:=Q(x)\partial_x^2 u +f(x)\partial_x u\,.
\label{eq:app-adjoint}
\end{align}
Note that $\mathcal L$ and $\mathcal L^\dagger$ have the same set of
eigenvalues.

We use the term \emph{eigenpair} to refer to an eigenvalue and eigenfunction of
the forward Fokker--Planck operator (e.g., $(\lambda,v)$ which satisfies
$\lambda v=\mathcal{L}v$) or the adjoint generator (e.g.,  $(\lambda,u)$ which
satisfies $\lambda u=\mathcal{L}^\dagger u$).  As shown in
Theorem~\ref{thm:cont-index}, the winding number can be equivalently computed
using either the forward or adjoint eigenfunction.

\begin{theorem}
\label{thm:cont-bound-fp}
For the Fokker--Planck equation~\eqref{eq:app-fp-cont}, every nonreal eigenpair
$(\lambda,u)$ has a well-defined winding number $\windFunc(u)\in
\{1,2,3,\dots\}$ and obeys
\begin{align}
\frac{|\imL|}{-\re\lambda}
\le
\frac{\Aff}{2\pi\windFunc(u)}
\le
\frac{\Aff}{2\pi}\,.
\label{eq:app-cont-bound-final}
\end{align}
The first bound is tight for uniform drift $f$ and diffusion $Q$.
\end{theorem}

\begin{proof}
We first prove the result for an eigenpair $(\lambda,u)$ of the adjoint
generator $\mathcal L^\dagger$. The eigenvalue equation reads
\begin{equation}
\lambda u=\mathcal L^\dagger u=Q(x)\partial_x^2u+f(x)\partial_xu\,.
\label{eq:app-adjoint-eigenvalue}
\end{equation}

For the fixed eigenfunction $u$, define
\[
\Gamma_\alpha(x):=\im(e^{\ii\alpha}u^*\partial_xu),\qquad \alpha\in\mathbb R.
\]
We first consider $\Gamma_0=\im(u^*\partial_xu)$. Multiplying
Eq.~\eqref{eq:app-adjoint-eigenvalue} by $u^*$,  using
$u^*\partial_x^2u=\partial_x(u^*\partial_xu)-|\partial_xu|^2$, and taking
imaginary parts gives
\begin{align}
Q\,\partial_x\Gamma_0+f\,\Gamma_0=(\imL)\,|u|^2\,.
\label{eq:app-gamma-ode}
\end{align}
This is a linear first-order ODE for $\Gamma_0$. Solving it over one period,
while using $\int_x^{x+1} f(z)/Q(z) \,dz=\Aff$ and $\Gamma_0(x+1)=\Gamma_0(x)$,
gives
\begin{align*}
(e^{\Aff}-1)\,\Gamma_0(x)
=(\imL)\int_x^{x+1}e^{\int_x^y f(z)/Q(z) \,dz}\,\frac{|u(y)|^2}{Q(y)}\,dy\,.
\end{align*}
The integral is strictly positive. Hence, if $\Aff=0$, no nonreal eigenvalue
exists. If $\Aff\ne0$ and $\imL\ne0$, then
$\sign\Gamma_0(x)=\sign[\imL/(e^\Aff-1)]\ne0$ for all $x$, thus $u(x)\ne0$ for
all $x$.

This allows us to introduce the polar representation
$u(x)=e^{\xi(x)+\ii\theta(x)}$, where $\xi$ is periodic and $\theta$ is a lifted
phase, so that $\partial_xu=(\partial_x\xi+\ii\partial_x\theta)u$ and
$\Gamma_0=|u|^2\partial_x\theta$. Since $\Gamma_0$ has a definite sign
everywhere, the lifted phase changes monotonically.

Without loss of generality, hereafter we assume that $\imL>0$. Since we assumed
$\Aff\ge 0$, this implies that  $\Gamma_0>0$ and $\partial_x\theta>0$
everywhere, hence the lifted phase is strictly increasing. Because $u$ is
periodic and nonvanishing, we have
\begin{align}
\windFunc(u)=\frac{1}{2\pi}\int_0^1\partial_x\theta(x)\,dx\in\{1,2,3,\dots\}\,.
\label{eq:app-cont-winding}
\end{align}

Let us write the eigenvalue as $\lambda=-|\lambda|e^{-\ii\psi}$; since
$\imL=|\lambda|\sin\psi>0$, we have $\psi\in(0,\pi)$. We now consider
$\Gamma_\psi$. Multiplying Eq.~\eqref{eq:app-adjoint-eigenvalue} by
$e^{\ii\psi}u^*$ and proceeding as above gives
\begin{align}
Q\,\partial_x\Gamma_\psi+f\,\Gamma_\psi=Q\sin\psi\,|\partial_xu|^2\,.
\label{eq:app-gamma-psi-ode}
\end{align}
Note that the term $\im(e^{\ii\psi}\lambda)|u|^2$ vanishes because
$e^{\ii\psi}\lambda=-|\lambda|$ is real, and the term
$Q\sin\psi\,|\partial_xu|^2$ arises from
$\im(e^{\ii\psi}u^*\partial_x^2u)=\partial_x\Gamma_\psi-\sin\psi\,|\partial_xu|%
^2$. We solve this ODE in the same way to write
\begin{align*}
(e^{\Aff}-1)\,\Gamma_\psi(x)=\sin\psi\int_x^{x+1}e^{\int_x^y f(z)/Q(z) \,
dz}\,|\partial_yu(y)|^2\,dy\,.
\end{align*}
The right-hand side is strictly positive, since $\sin\psi>0$ and
$\partial_xu\not\equiv0$ (otherwise $u$ would be constant and $\lambda=0$).
Hence $\Gamma_\psi(x)>0$ for all $x$.

Next, we introduce $D_\psi:={\Gamma_\psi}/{\vert u\vert^2}>0$, which obeys
\begin{align*}
D_\psi=(\partial_x\xi)\sin\psi+(\partial_x\theta)\cos\psi\,.
\end{align*}
Since $\xi$ is periodic, $\partial_x\xi$ cannot be positive everywhere. Since
$\partial_x\theta>0$,  $D_\psi>0$ forces $\cos\psi>0$. Thus $\psi\in(0,\pi/2)$,
$\reL=-|\lambda|\cos\psi<0$, and $\tan\psi=\imL/(-\reL)$.

We divide Eq.~\eqref{eq:app-gamma-psi-ode} by $Q\Gamma_\psi$ and integrate over
the circle,
\begin{align}
\Aff=\sin\psi\int_0^1\frac{|\partial_xu|^2}{\Gamma_\psi}\,dx
=\sin\psi\int_0^1\frac{(\partial_x\xi)^2+(\partial_x\theta)^2}{D_\psi}\,dx\,,
\label{eq:alt-aff-identity}
\end{align}
where we used that $\ln\Gamma_\psi$ is periodic and
$|\partial_xu|^2=[(\partial_x\xi)^2+(\partial_x\theta)^2]\,|u|^2$. Using
$(\partial_x\theta)\cos\psi=D_\psi-(\partial_x\xi)\sin\psi$, we write
\begin{align*}
&\sin\psi\,\bigl[(\partial_x\xi)^2+(\partial_x\theta)^2\bigr]\\
&\quad=\frac{\sin\psi}{\cos^2\psi}\Bigl[(\partial_x\xi)^2\cos^2\psi+%
\bigl(D_\psi-(\partial_x\xi)\sin\psi\bigr)^2\Bigr]\\
&\quad=\frac{\sin\psi}{\cos^2\psi}\Bigl[(\partial_x\xi)^2+D_\psi^2\Bigr]%
-2(\partial_x\xi)\tan^2\psi\,D_\psi\\
&\quad = \frac{\sin\psi}{\cos^2\psi}(\partial_x\xi)^2+D_\psi
(\partial_x\theta)\tan\psi -(\partial_x\xi)\tan^2\psi\,D_\psi.
\end{align*}
Plugging into Eq.~\eqref{eq:alt-aff-identity} and integrating
over the circle, while using $\int_0^1\partial_x\theta\,dx=2\pi\windFunc(u)$ and
$\int_0^1\partial_x\xi\,dx=0$, gives
\begin{align}
\Aff-2\pi\windFunc(u)\tan\psi
=\frac{\sin\psi}{\cos^2\psi}\int_0^1\frac{(\partial_x\xi)^2}{D_\psi}\,dx\ \ge
0\,,
\label{eq:app-excess}
\end{align}
with equality if and only if $|u|$ is constant. Since $\tan\psi=\imL/(-\reL)$,
Eq.~\eqref{eq:app-excess} is the first inequality in
Eq.~\eqref{eq:app-cont-bound-final}. The second inequality follows from
$\windFunc(u)\ge1$.

For constant $f$ and $Q$, the eigenfunctions are $u_m(x)=e^{2\pi\ii mx}$ with
eigenvalues $\lambda_m=-Q(2\pi m)^2+\ii f(2\pi m)$. Then $\Aff=f/Q$,
$\windFunc(u_m)=|m|$, and $|u_m|$ is constant, so the first inequality in
Eq.~\eqref{eq:app-cont-bound-final} is saturated.

The result for a forward eigenpair $(\lambda,v)$ follows from
Theorem~\ref{thm:cont-index} below: a corresponding adjoint eigenfunction $u$
has $\windFunc(v)=\windFunc(u)$, so the bound just proved applies also to $v$.
\end{proof}

We now prove the relationship between winding numbers and eigenvalue indices,
analogous to Theorem~\ref{thm:ordered-winding}.

\begin{theorem}
\label{thm:cont-index}
Sort the eigenvalues of the Fokker--Planck operator~\eqref{eq:app-fp-cont}  by
real part as
$0=\lambda_1\ge \re\lambda_2\ge \re\lambda_3\ge \cdots$ (counting algebraic
multiplicity). For any nonreal  $\lambda_k$, consider a corresponding
eigenfunction $\mathcal L v=\lambda_k v$ and adjoint eigenfunction  $\mathcal
L^\dagger u=\lambda_k u$. Then, their winding numbers obey
\begin{align}
\windFunc(u)=\windFunc(v)=\lfloor{k}/{2}\rfloor\,.
\end{align}
\end{theorem}
\begin{proof}
First consider an adjoint eigenpair $(\lambda_k,u)$. The adjoint part of the
proof of Theorem~\ref{thm:cont-bound-fp} shows that $u$ is nonvanishing and
has a strictly monotone lifted phase $\theta$. Consequently, the nonzero real
function
\[
\evblockfunc(x):=\re u(x)=|u(x)|\cos\theta(x)
\]
has exactly $2\windFunc(u)$ simple zeros on the circle.

In the following, we will apply the results of Angenent and
Fiedler~\cite{angenent1988dynamics}. We note that, although the results of that
paper were originally stated for analytic coefficients, they were later
generalized to smooth coefficients~\cite{angenent1988zero}.
These results also assume unit diffusion. If this is not the case, set
\[
y=\Phi(x):=\kappa\int_0^x Q(z)^{-1/2}\,dz,
\quad \kappa^{-1}=\int_0^1 Q(z)^{-1/2}\,dz.
\]
This smooth increasing bijection of the unit circle induces the similar operator
\begin{align*}
\mathcal L^\prime g
:=\bigl[\mathcal
L^\dagger(g\circ\Phi)\bigr]\circ\Phi^{-1}=\kappa^2\partial_y^2g+a(y)\partial_yg,
\end{align*}
where $a:=[\kappa(f-\tfrac12\partial_xQ)/\sqrt{Q}]\circ\Phi^{-1}$ is smooth
and periodic. This similarity maps each eigenfunction $u$ of $\mathcal
L^\dagger$, $\mathcal L^\dagger u= \lambda u$, to the eigenfunction
$g(y)=u(\Phi^{-1}(y))$ (hence $u(x)=g(\Phi(x))$)  of $\mathcal L^\prime$,
$\mathcal{L}^\prime g = \lambda g$, while preserving  algebraic multiplicities
of the eigenvalues.
The coordinate change also preserves zero crossings of real-valued functions
and the total phase advance, hence the winding number.
Rescaling time as $s=\kappa^2t$ then gives the unit-diffusion generator
$\kappa^{-2}\mathcal L^\prime$, with eigenvalues $\lambda/\kappa^2$ and
unchanged eigenfunctions. Both steps preserve real-part ordering, algebraic
multiplicities, and winding. Therefore, in the following we assume $Q=1$ without
loss of generality.

Let $T_t:=e^{t\mathcal L^\dagger}$
be the time-$t$ evolution operator generated by  $\mathcal{L}^\dagger$.
For any $t>0$, $T_t$ is compact~\cite{henry2006geometric}. Then, by the Spectral
Mapping Theorem, there is a one-to-one mapping between eigenpairs $(\lambda,u)$
of $\mathcal{L}^\dagger$ and eigenpairs $(e^{t\lambda},u)$ of $T_t$ for a
generic $t>0$ that avoids exponential collisions.

Theorem~2.1 in Ref.~\cite{angenent1988dynamics} states that the eigenvalues of
$T_t$ can be sorted by modulus as
\[
|e^{t\lambda_1}|>|e^{t\lambda_2}|\ge|e^{t\lambda_3}|>|e^{t\lambda_4}|\ge|%
e^{t\lambda_5}|>\cdots ,
\]
thus falling into strictly ordered blocks of two. Since
$|e^{t\lambda_j}|=e^{t\re \lambda_j}$, this is equivalent to the block ordering
of the generator eigenvalues as
\begin{align*}
     \lambda_1 > \re \lambda_2 \ge \re \lambda_3 > \re \lambda_4 \ge\re
\lambda_5 > \dots\,.
\end{align*}
Next, consider two eigenfunctions $g,h$ of $T_t$ associated with eigenvalues
$e^{t\lambda_{2\evblockix}}$ and $e^{t\lambda_{2\evblockix+1}}$, thus also
eigenfunctions of $\mathcal{L}^\dagger$ associated with eigenvalues
$\lambda_{2\evblockix},\lambda_{2\evblockix+1}$. Theorem~2.2 in
Ref.~\cite{angenent1988dynamics} implies that any nonzero real-valued function
in $\operatorname{span}_\mathbb{C}\{g,h\}$ has exactly \(2\evblockix\) zero
crossings.

Consider an eigenpair $(\lambda_k, u)$ of the adjoint operator \(\mathcal
L^\dagger\).  Since $\mathcal{L}^\dagger$ has real coefficients, $u^*$ is also
an eigenfunction with eigenvalue \(\lambda_k^*\). Therefore, this pair of
eigenvalues belongs to block $\{\lambda_{2\evblockix},\lambda_{2\evblockix+1}\}$
with $\evblockix=\lfloor k/2\rfloor$. The set of real-valued
functions in  $\operatorname{span}_\mathbb{C}\{u,u^*\}$ is equal to
$\operatorname{span}_\mathbb{R}\{\re  u,\im  u\}$.  Since $\evblockfunc \in
\operatorname{span}_\mathbb{R}\{\re u,\im u\}$, it   has \(2\lfloor k/2\rfloor\)
zero crossings.
Combining these results implies  $\windFunc(u)=\lfloor k/2\rfloor$.

It remains to consider an eigenpair $(\lambda_k,v)$ of the forward operator
$\mathcal L$. Let $\rhoss>0$ be the stationary density (Lemma~\ref{lem:fp-p-pos}
below). In addition, define the auxiliary adjoint operator
$\widetilde{\mathcal L}^{\dagger}g:=\rhoss^{-1}\mathcal L(\rhoss g)$. Using the
product rule, it  can be expressed as
\begin{align*}
\widetilde{\mathcal L}^{\dagger}g
& =Q\partial_x^2 g+\left[2\frac{\partial_x(Q\rhoss)}{\rhoss}-f\right]\partial_x
g + \mathcal L(\rhoss) g/\rhoss\,.
\end{align*}
The last term vanishes by stationarity, $\mathcal L(\rhoss)=0$.
As a result, the operator $\widetilde{\mathcal L}^{\dagger}$ has the same form
as Eq.~\eqref{eq:app-adjoint}, with a smooth periodic drift and diffusion $Q>0$.

Next, observe that $h:=v/\rhoss$ satisfies $\widetilde{\mathcal
L}^{\dagger}h=\lambda_k h$. Since $\widetilde{\mathcal L}^{\dagger}$ is similar
to $\mathcal L$, they have the same eigenvalues (counting algebraic
multiplicity), so $\lambda_k$ has the same ordered index for both.
We now apply the adjoint derivation  above to $h$, which shows that $h$ is
nonvanishing and $\windFunc(h)=\lfloor k/2\rfloor$. Since $v=\rhoss h$ and
$\rhoss>0$, the two functions have the same phase, so
$\im[(\partial_xv)/v]=\im[(\partial_xh)/h]$ and
$\windFunc(v)=\windFunc(h)=\lfloor k/2\rfloor$.

\end{proof}

We finish with a lemma that was used above.

\begin{lemma}
\label{lem:fp-p-pos}
The stationary distribution of the Fokker--Planck
equation~\eqref{eq:app-fp-cont} is strictly positive.
\end{lemma}
\begin{proof}
Let
\(\rhoss\) be the stationary distribution and $\jss=f\rhoss-\partial_x(Q\rhoss)$
the
stationary flux. The steady-state condition $0=-\partial_x \jss$ means that
$\jss$ is constant on the cycle.

Next, we rewrite $\jss=f\rhoss-\partial_x(Q\rhoss)$ as
\begin{align}
\partial_x\rhoss=\frac{f-\partial_xQ}{Q}\rhoss-\frac{\jss}{Q}\,.
\label{eq:zxsa93}
\end{align}
Suppose \(\rhoss(x_0)=0\) for some \(x_0\). Since \(\rhoss\ge0\), \(x_0\) is a
minimum of \(\rhoss\), so
\(\partial_x\rhoss(x_0)=0\). Plugging into Eq.~\eqref{eq:zxsa93} gives
\(\jss=0\). In this case, $\rhoss(x)=0$ is the unique solution to the
first-order ODE~\eqref{eq:zxsa93} with initial condition $\rhoss(x_0)=0$,
contradicting
\(\int_0^1\rhoss\,dx=1\). Therefore, \(\rhoss(x)>0\) everywhere.

\end{proof}

\section{PROOF OF THEOREM~\ref{thm:multicyclic-main}}
\label{app:proof-main}

\begin{reptheoremnum}{thm:multicyclic-main}
\mainTheoremStmt{}
\end{reptheoremnum}

The proof proceeds via the following steps:
\begin{enumerate}[itemsep=2pt,leftmargin=*,label=\arabic*),widest={SM3.5}]
\item[SM\ref{app:proofsetup})] We define an $n\times n$ matrix $M$ in terms of
$\lambda$, $\bm u$, and $\mixc$. We then use it to construct a weighted directed
graph $G$ based on the generator $\W$ and matrix $M$.

\item[SM\ref{app:esconstruct})] We apply Theorem~\ref{thm:witness-cycle-m4}, a
graph-theoretic result proved below, to extract a certifying cycle $c$ from
graph $G$.

\item[SM\ref{app:cyclewinding})] We show that the certifying cycle $c$ is
admissible and compute its winding number $\windFunc_c(\bm u)$.
\item[SM\ref{app:finalbound})] We lower-bound the cycle affinity $\Aff_c$ as a
sum of $n_c$ terms, each one being a convex function of the log-amplitude and
phase of the eigenvector ratio $u_\yy/u_\xx$ on edges of $c$. Applying Jensen's
inequality then gives the final result.
\end{enumerate}

\subsection{Assumptions}

Without loss of generality, we make three assumptions.

First, we assume that $n \ge 3$, since generators with $n=2$ have only
real-valued eigenvalues.

Second, we assume that the nonreal eigenvalue $\lambda$ has $\imL > 0$. To show
the result for a nonreal left eigenpair $(\lambda,\bm u)$ with $\imL < 0$, one
may consider the conjugate eigenpair $(\lambda^*,\bm u^*)$ with $\imLconj
=-\imL> 0$. It is straightforward to verify that the winding
number~\eqref{eq:winding-multicycle} is invariant under conjugation, so
$\windFunc_c(\bm u)=\windFunc_c(\bm u^*)$ for any admissible cycle where they
are well-defined. Therefore, if the theorem holds for eigenpairs with $\imL >0$,
it also holds for eigenpairs with $\imL <0$.

Finally, we assume that $\bm u$ is a left eigenvector ($\W^\top \bm u=\lambda
\bm u$). The case where $\bm u$ is a right eigenvector is addressed separately
in SM\ref{app:righteig} at the end of this proof.

\subsection{Setup and graph construction}
\label{app:proofsetup}
For convenience, we define the constants
\begin{align}
\begin{aligned}
\alpha :=
\mixc\frac{\imL}{\minusReL} \qquad
\beta :=(1-\mixc) \frac{\imL}{\dualReL} \,.
\end{aligned}
\label{eq:asdfzz0}
\end{align}
A simple application of Gershgorin's Circle Theorem shows that any nonreal
eigenvalue obeys $-2\maxrate < \reL< 0$. Therefore, $\alpha \ge 0$ and $\beta
\ge 0$, with at least one of these inequalities being strict. Thus, we have the
inequalities
\begin{align}
\alpha +\beta > 0 \qquad \alpha +\beta \ge \vert \alpha -\beta \vert \,.
\label{eq:alphabetineqs}
\end{align}

In addition,
using the eigenvector $\bm u$, we define the following matrix of edge-wise
observables:
\begin{align}
M_{\yy\xx}:= \operatorname{Im}(u_\yy u_\xx^*)
+(\alpha\!-\!\beta)\operatorname{Re}(u_\yy u_\xx^*)
-(\alpha\!+\!\beta)|u_\xx|^2 .
\label{eq:Mdef}
\end{align}

Define a weighted directed graph \(G=(V,E,\bm f)\) with
\begin{align*}
V&:=\{\xx \in \{1,\dots,n\}: u_\xx \neq 0\}\\
E&:=\{\yxedge:\xx,\yy\in V,\ \W_{\yy\xx}>0,\ u_\yy\ne u_\xx,u_\yy\ne -u_\xx\}\\
\fyxedge&:=\W_{\yyxx}M_{\yyxx},\qquad\qquad \yxedge\in E.
\end{align*}
We call
\(\yxedge\in E\) a \emph{nonnegative edge} if \(\fyxedge\ge0\).

\newcommand\sinkSCC{\mathcal{S}}

\subsection{Cycle extraction}
\label{app:esconstruct}
\label{app:subgraphextract}

We will apply Theorem~\ref{thm:witness-cycle-m4}, a graph-theoretic result
proved in SM\ref{app:cycle-extraction} below, to the weighted graph $G$. We
first verify that this graph satisfies the four assumptions of the theorem.

First, since \(\bm u\) is a nonzero eigenvector, \(V\) is not empty.

Second, every vertex has nonnegative total outgoing weight. For each \(\xx\in
V\), we have
\begin{align*}
\sum_{\yy:\yy\ne \xx} \W_{\yyxx}M_{\yyxx}
&=
\bigl[\imL + (\alpha-\beta)\reL + 2\beta\W_{\xx\xx}\bigr]|u_\xx|^2 \,,
\end{align*}
where we used \(\sum_{\yy:\yy\ne \xx} \W_{\yyxx} u_\yy=(\lambda
-\W_{\xx\xx})u_\xx\) and \(\sum_{\yy:\yy\ne \xx} \W_{\yyxx} =-\W_{\xx\xx}\).
Plugging in the definitions of $\alpha,\beta$ and rearranging gives
\begin{align*}
\sum_{\yy:\yy\ne \xx} \W_{\yyxx}M_{\yyxx}
&=
2\beta(\maxrate+\W_{\xx\xx})
|u_\xx|^2
\ge 0 \,.
\end{align*}
Therefore, for each \(\xx\in V\),
\[
\sum_{\substack{\yy \in V:\\ \yxedge\in E}} \fyxedge
=
\sum_{\substack{\yy \in V:\\ \yxedge\in E}} \W_{\yyxx}M_{\yyxx}
\ge
\sum_{\yy:\yy\ne \xx} \W_{\yyxx}M_{\yyxx}
\ge 0.
\]
The first inequality holds because all additional terms are nonpositive:
nonphysical edges ($\W_{\yyxx}=0$) contribute nothing,
physical edges with \(\yy\notin V\) contribute
\(\W_{\yyxx}M_{\yyxx}=-\W_{\yyxx}(\alpha+\beta)|u_\xx|^2<0\), and physical edges
with
\(u_\yy=\pm u_\xx\) contribute \(\W_{\yyxx}M_{\yyxx}\le0\).

Third, every vertex has an outgoing
edge in $E$. To show this, for any \(\xx\in V\), we write the eigenvector
equation as
\[
\lambda
=\sum_{\yy:\yy\ne\xx}\W_{\yyxx}(u_\yy/u_\xx-1).
\]
Now suppose for contradiction that no outgoing edge in
\(E\) leaves \(\xx\). Then every physical edge out of \(\xx\) either goes
to \(\yy\notin V\) (so $u_\yy = 0$) or satisfies \(u_\yy=\pm u_\xx\). In these
cases, the right-hand side sums to a real number, contradicting the assumption
that $\lambda$ is nonreal.

Fourth, every reversible nonnegative edge has
strictly negative reverse weight. For any edge \(\yxedge\in E\), we have
\begin{align}
M_{\yyxx}+M_{\xxyy}
=-\alpha|u_\yy-u_\xx|^2-\beta|u_\yy+u_\xx|^2 < 0\,.
\label{eq:Msymm}
\end{align}
For a nonnegative edge \(\fyxedge\ge0\), this gives the implications
\[
\fyxedge\ge0 \implies \W_{\yyxx}M_{\yyxx}\ge0 \implies M_{\yyxx}\ge0
\implies M_{\xxyy}<0 .
\]
For a reversible edge, \(\revedge\in E\) and
\(\fxyedge=\W_{\xxyy}M_{\xxyy}<0\).

We now apply Theorem~\ref{thm:witness-cycle-m4}
to
\(G=(V,E,\bm f)\).
The theorem returns a cycle $c$ consisting of nonnegative edges $\fyxedge \ge
0$. This cycle satisfies one of two possible conditions.
First, \(c\) may contain an irreversible edge, in which case
\(\Aff_c=\infty\). Otherwise, every edge of \(c\) is reversible and the theorem
gives
\begin{align}
\chi_c
&:= \sum_{\yxedge \in c} \ln\frac{\fyxedge}{-\fxyedge}\nonumber\\
&=\sum_{\yxedge \in c} \ln\frac{\W_{\yyxx}M_{\yyxx}}{-\W_{\xxyy}M_{\xxyy}}
\nonumber \\
&= \Aff_{c} + \sum_{\yxedge \in c} \ln\frac{M_{\yyxx}}{-M_{\xxyy}}\ge 0 \,.
\label{eq:chires}
\end{align}

\subsection{Cycle winding}
\label{app:cyclewinding}

For convenience, we define the constants
\begin{align*}
\psi&:=\myarctan(\alpha-\beta)\in(-\pi/2,\pi/2)\,,\\
\sigma&:=\frac{\alpha+\beta}{\sqrt{1+(\alpha-\beta)^2}} > 0\,.
\end{align*}
We will also use the inequality
\begin{align}
\sigma &\ge \frac{\vert \alpha-\beta\vert}{\sqrt{1+(\alpha-\beta)^2}}=\vert \sin
\psi\vert\,.
\label{eq:appsigineq}
\end{align}

Since $c\subseteq V$, all states $\xx$ in the cycle $c$ have $u_\xx\ne 0$ and
all edges $\yxedge\in c$ have well-defined ratios $u_\yy/u_\xx$.
We write these ratios in polar coordinates as
\begin{align}
\frac{u_\yy}{u_\xx}=e^{-\xlog_{\yyxx}+\ii(\varphi_{\yyxx}-\psi)}\,.
\label{eq:polar}
\end{align}

Using the polar coordinates, we may write Eq.~\eqref{eq:Mdef} as
\begin{align}
M_{\yyxx}
&=
|u_\xx|^2\sqrt{1+(\alpha-\beta)^2}(e^{-\xlog_{\yyxx}}\sin\varphi_{\yyxx}%
-\sigma)\,.
\label{eq:local-M-00}
\end{align}
Recall that \(\fyxedge\ge0\) and \(\W_{\yyxx}>0\) on every edge \(\yxedge\in
c\). Therefore, \(M_{\yyxx}\ge0\) and so
\begin{align}
0<\sigma\le e^{-\xlog_{\yyxx}}\sin\varphi_{\yyxx}\,,
\label{eq:xzcvs}
\end{align}
implying that $\sin\varphi_{\yyxx}>0$. Thus, we may choose $\varphi_{\yyxx} \in
(0,\pi)$.

For convenience, we define the \emph{winding phase} as
\begin{align}
\vartheta:=\frac{1}{n_c}\sum_{\yxedge\in c}(\varphi_{\yyxx}-\psi)\,.
\label{eq:vtdef}
\end{align}
Below, we will use the bounds
\begin{align}
\vartheta+\psi \in(0,\pi) \,,
\label{eq:vtineq}
\end{align}
which follow since $\vartheta+\psi=\frac{1}{n_c}\sum_{\yxedge\in
c}\varphi_{\yyxx}$ and each $\varphi_{\yyxx}\in (0,\pi)$. In addition, the
identity $\prod_{\yxedge \in c} (u_\yy/u_\xx)=e^{-\sum\xlog_{\yyxx}+\ii n_c
\vartheta}=1$ implies
\begin{align}
\sum_{\yxedge\in c} \xlog_{\yyxx} = 0\,,\qquad
\frac{n_c\vartheta}{2\pi} \in\mathbb Z\,.
\label{eq:cyclicid}
\end{align}

We take logarithms in Eq.~\eqref{eq:xzcvs}, average across the cycle, and
simplify to
derive the bound
\begin{align}
\label{eq:zvxs02}
\ln\sigma
\le
\frac{1}{n_c}\sum_{\yxedge\in c}\ln\sin\varphi_{\yyxx}\,.
\end{align}
Since $\ln \sin$ is concave on $(0,\pi)$, Jensen's inequality implies
\begin{align}
\label{eq:appjens}
 \frac{1}{n_c}\sum_{\yxedge\in c}\ln\sin\varphi_{\yyxx}\le
\ln\sin(\vartheta+\psi)\,.
\end{align}
Combining the last two inequalities gives
\begin{align}
\ln\sigma
\le
\ln\sin(\vartheta+\psi)\,.
\label{eq:geometry-bound-m4}
\end{align}

We now prove that \(\vartheta\in(0,\pi)\) by contradiction.
Suppose for the moment that \(\vartheta\le0\), in which case
Eq.~\eqref{eq:vtineq} implies that
\(\psi>0\). We combine the inequalities~\eqref{eq:appsigineq}
and~\eqref{eq:geometry-bound-m4}, along with the fact that \(\sin\) is
monotonically increasing on \((0,\pi/2)\), to write
\begin{align}
\sigma \le \sin(\vartheta+\psi)\le \sin\psi \le \sigma \,,
\label{eq:fdd01}
\end{align}
which forces $\vartheta = 0$. Now suppose that \(\vartheta\ge \pi\), in which
case
\(\psi<0\). Since \(\sin\) is decreasing on \((\pi/2,\pi)\), we may write
\begin{align}
\sigma\le \sin(\vartheta+\psi)\le \sin(\pi+\psi)=-\sin\psi\le\sigma \,,
\label{eq:fdd02}
\end{align}
which forces $\vartheta = \pi$. In either case, equality holds throughout
Eqs.~\eqref{eq:zvxs02}--\eqref{eq:appjens}. Strict concavity of $\ln\sin$ on
$(0,\pi)$ then forces $\varphi_{\yyxx}=\vartheta+\psi$ on every edge. Since
$\sigma=\sin(\vartheta+\psi)$, Eq.~\eqref{eq:xzcvs} gives $\xlog_{\yyxx}\le0$,
and $\sum_{\yxedge\in c}\xlog_{\yyxx}=0$ then forces $\xlog_{\yyxx}=0$ on every
edge. Thus Eq.~\eqref{eq:polar} gives $u_\yy/u_\xx=e^{\ii\vartheta}$, which
equals $1$ when $\vartheta=0$ and $-1$ when $\vartheta=\pi$. Both are excluded
by $c\subseteq E$, proving $\vartheta\in(0,\pi)$.

Combining $\vartheta\in(0,\pi)$ with $n_c\vartheta/2\pi \in\mathbb Z$ implies
that
\begin{align}
\label{eq:vtnd1}
\frac{n_c\vartheta}{2\pi} \in \{1,\dots, \lceil n_c/2\rceil - 1\}\,.
\end{align}

Finally, we connect the winding phase \(\vartheta\) to the winding number
\(\windFunc_c(\bm u)\) from Eq.~\eqref{eq:winding-multicycle}.
Since $\varphi_{\yyxx}-\psi\in(-\psi,\pi-\psi)$, the phases of the edge ratios
$u_\yy/u_\xx=e^{-\xlog_{\yyxx}+\ii(\varphi_{\yyxx}-\psi)}$ belong to a common
open interval of length $\pi$.
Thus, the cycle \(c\) is admissible for \(\bm u\).
To compute the
winding number, choose the branch parameter using the average phase vector,
\begin{align}
\branchParam
:=
\Arg\sum_{\yxedge\in c}
e^{\ii\Arg(u_\yy/u_\xx)} .
\label{eq:branchParamapp}
\end{align}
Because the edge-ratio phases belong to a common open interval of length
\(\pi\), this average phase vector $\sum_{\yxedge\in c}
e^{\ii\Arg(u_\yy/u_\xx)}$  is nonzero and its phase lies within the same
interval.
Then, $\arg_{\branchParam} z:=\Arg(e^{-\ii \branchParam}z)+\branchParam$ places
the branch cut at the opposite side $\branchParam\pm \pi$, guaranteeing that it
does not intersect this interval.

Since all edge-ratio phases lie in a common interval and the branch cut of
$\arg_{\branchParam}$ avoids this interval,
$\arg_{\branchParam}({u_\yy}/{u_\xx})$ differs from $\varphi_{\yy\xx}-\psi$ by a
constant integer multiple of \(2\pi\) on all edge ratios. This implies that
\begin{align*}
\windFunc_c(\bm u)&=\Bigg\vert
\frac{1}{2\pi}\sum_{\yxedge\in c}
\arg_{\branchParam}\frac{u_\yy}{u_\xx}
\Bigg\vert_{n_c}
\\
&=\Bigg\vert
\frac{1}{2\pi}\sum_{\yxedge\in c}(\varphi_{\yy\xx}-\psi)
\Bigg\vert_{n_c}\\
&=
\Big\vert\frac{n_c\vartheta}{2\pi}\Big\vert_{n_c}=\frac{n_c\vartheta}{2\pi}\,.
\end{align*}
In the last line, we used Eq.~\eqref{eq:vtnd1} along with the fact that $\vert
m\vert_{n_c}=m$ for any $m\in \{1,\dots, \lceil n_c/2\rceil - 1\}$.

\subsection{Final bound}
\label{app:finalbound}

For convenience, we define a constant to represent the left side of the
inequality in
Theorem~\ref{thm:multicyclic-main},
\begin{align}
X:=\alpha\tan\frac{\vartheta}{2}
+
\beta\cot\frac{\vartheta}{2} \,.
\label{eq:appXdef}
\end{align}
Here we used the definitions of $\alpha,\beta$ from Eq.~\eqref{eq:asdfzz0} and
\(\vartheta=2\pi \windFunc_c(\bm u)/n_c\in(0,\pi)\). To prove
Theorem~\ref{thm:multicyclic-main}, we show that
\begin{align}\label{eq:trg0}
X\le \tanh\frac{\Aff_c}{2n_c}.
\end{align}
Using standard trigonometric identities, we may write
\begin{align}
X
&=\alpha\frac{1-\cos\vartheta}{\sin\vartheta}
+\beta\frac{1+\cos\vartheta}{\sin\vartheta}\nonumber \\
&=\frac{(\alpha+\beta)-(\alpha-\beta)\cos\vartheta}
{\sin\vartheta} \nonumber \\
&=\frac{\sigma-\sin\psi\cos\vartheta}
{\cos\psi\sin\vartheta} \label{eq:bzxvsdf2}\\
&=1+\frac{\sigma-\sin(\vartheta+\psi)}
{\cos\psi\sin\vartheta}\,.\label{eq:bzxvsdf3}
\end{align}
Here we used the definitions of $\sigma,\psi$ and the sine addition identity.
The denominator $\cos\psi\sin\vartheta$ is positive, since
$\psi\in(-\pi/2,\pi/2)$ and $\vartheta\in(0,\pi)$.

We first consider the irreversible case where \(\Aff_c=\infty\), so
\begin{align}
\tanh\frac{\Aff_c}{2n_c}=1\,.
\end{align}
Combining~\eqref{eq:geometry-bound-m4} and~\eqref{eq:bzxvsdf3} gives $X\le1$,
proving~\eqref{eq:trg0}.

Next, we consider the case where \(\Aff_c < \infty\).
On any edge of the cycle, Eq.~\eqref{eq:local-M-00} implies that
\begin{align*}
\ln \frac{M_{\yyxx}}{-M_{\xxyy}}
&=
\ln \frac{|u_\xx|^2\left(e^{-\xlog_{\yyxx}}\sin\varphi_{\yyxx}-\sigma\right)}
{|u_\yy|^2\left(e^{\xlog_{\yyxx}}\sin(\varphi_{\yyxx}-2\psi)+\sigma\right)}
\\
&=-J_{\sigma,\psi}(\varphi_{\yyxx},\xlog_{\yyxx}) +
\ln\frac{|u_\xx|^2}{|u_\yy|^2}\,,
\end{align*}
where we introduced the function
\begin{align*}
J_{\sigma,\psi}(\varphi,\xlog)
&:= \ln\frac{e^{\xlog}\sin(\varphi-2\psi)+\sigma}
{e^{-\xlog}\sin\varphi-\sigma}
\end{align*}
on the domain $\jDomSG
:=
\{(\varphi,\xlog):
0<\varphi<\pi,\,
\xlog<\ln\frac{\sin\varphi}{\sigma}
\}$.
Combining with Eq.~\eqref{eq:chires} gives
\begin{align}
\Aff_c &\ge \sum_{\yxedge\in c} J_{\sigma,\psi}(\varphi_{\yyxx},\xlog_{\yyxx}) -
\sum_{\yxedge\in c} \ln\frac{|u_\xx|^2}{|u_\yy|^2}\nonumber \\
&= \sum_{\yxedge\in c} J_{\sigma,\psi}(\varphi_{\yyxx},\xlog_{\yyxx})\,.
\label{eq:affbound9}
\end{align}

Observe that all $\varphi_{\yyxx} \in (0,\pi)$. Moreover, whenever \(\Aff_c <
\infty\),
the certifying cycle in Theorem~\ref{thm:witness-cycle-m4} obeys
\(\fyxedge>0\) on every edge of \(c\), hence \(M_{\yyxx}>0\).
Equation~\eqref{eq:local-M-00} then implies that
$\xlog_{\yyxx}<\ln({\sin\varphi_{\yyxx}}/{\sigma})$, so all points
\((\varphi_{\yyxx},\xlog_{\yyxx})\) lie in the open convex domain $\jDomSG$. By
convexity, their average also lies in this domain, so $\sigma <\sin(\vartheta
+\psi)$.
Next, we apply Theorem~\ref{thm:local-convexity}, which proves that
\(J_{\sigma,\psi}\) is convex on its domain.
Applying Jensen's inequality to Eq.~\eqref{eq:affbound9} gives
\[
\Aff_c \ge n_c J_{\sigma,\psi}(\vartheta+\psi,0)
=
n_c \ln\frac{\sin(\vartheta-\psi)+\sigma}
{\sin(\vartheta+\psi)-\sigma}\,,
\]
where we used \(\sum_{\yxedge\in c} \xlog_{\yyxx}=0\) and
\(n_c^{-1}\sum_{\yxedge\in c}\varphi_{\yyxx}=\vartheta+\psi\).
Finally, rearranging gives the bound~\eqref{eq:trg0} via
\begin{align*}
\tanh\frac{\Aff_c}{2n_c}
&\ge
\tanh\left[
\frac12\ln\frac{\sin(\vartheta-\psi)+\sigma}
{\sin(\vartheta+\psi)-\sigma}
\right]\\
&=
\frac{2\sigma+\sin(\vartheta-\psi)-\sin(\vartheta+\psi)}
{\sin(\vartheta-\psi)+\sin(\vartheta+\psi)} \\
&=
\frac{\sigma-\sin\psi\cos\vartheta}
{\cos\psi\sin\vartheta} =X \,.
\end{align*}
In the last line, we used the sine sum and difference identities and then
Eq.~\eqref{eq:bzxvsdf2}.

\newcommand{\SCC}{\mathcal{S}}
\newcommand{\Wscc}{\W^{\SCC}}
\newcommand{\perrev}{\lambda_1^{\SCC}}
\newcommand{\perxx}{a}
\newcommand{\perron}{\bm{\perxx}}
\newcommand{\Dperron}{D}
\newcommand{\rvecS}{\bm u^\SCC}
\newcommand{\rvecSxx}{u^\SCC}
\newcommand{\Wadj}{ {\widetilde{\W}^{\SCC}} }
\newcommand{\evecAdj}{\widetilde{\bm u}}
\newcommand{\evecAdjxx}{\widetilde{u}}
\newcommand{\adjAff}{\widetilde{\Aff}}
\newcommand{\ccc}{{\widetilde c}}
\newcommand{\mccc}{\widetilde \winding_\ccc}

\subsection{Right eigenvector}
\label{app:righteig}
We now consider the case where $\bm u$ is a right eigenvector, so $\W \bm u
=\lambda \bm u$. We prove the result by reducing to the case of a left
eigenvector. We will use the concept of a \emph{strongly connected
component} (SCC) of a directed graph, defined as a maximal subset of
vertices such that every vertex is reachable from every other vertex by a
directed path.

We proceed via three steps.

\vspace{10pt}
\noindent\emph{Step 1: Restriction to an SCC.}
Let $\SCC_1,\dots,\SCC_K$ denote the SCCs of the directed graph of $\W$.
Index the states so that $\W$ is block-lower-triangular with diagonal blocks
$\{\W\vert_{\SCC_k,\SCC_k}\}$ in topological order: if $\W$ has a transition from $\SCC_j$
to $\SCC_i$ with $i\ne j$, then $j<i$.

Let the notation $\bm u \vert_{\SCC}$ indicate restriction of vector $\bm u$ to
the subset of states $\SCC$, and the notation $\W\vert_{\SCC,\SCC^\prime}$
indicate the restriction of matrix $\W$ to block $\SCC,\SCC^\prime$.
Let $k$ be the smallest index such that $\bm u|_{\SCC_k}\ne\zz$ (such a $k$
exists because $\bm u\ne\zz$). Writing $\W\bm u=\lambda\bm u$ as a
block-by-block identity, the $k$-th block-row reads
\begin{align}
 \lambda\,\bm u|_{\SCC_k}=\sum_{j\le k}\W\vert_{\SCC_k, \SCC_j}\,\bm
u|_{\SCC_j}= \W\vert_{\SCC_k, \SCC_k}\,\bm u|_{\SCC_k}.
 \label{eq:SCC-block-row}
\end{align}
First, we use that only terms with $j\le k$ appear because $\W$ is
block-lower-triangular, then we use that $\bm u|_{\SCC_j}=0$ for $j<k$.
In this way, we have shown that $\rvecS := \bm u|_{\SCC_k}$ is a right
eigenvector of $\Wscc:=\W\vert_{\SCC_k,\SCC_k}$ with eigenvalue $\lambda$,
\begin{align}
\Wscc\,\rvecS = \lambda\,\rvecS, \qquad \rvecS \ne \zz\,.
\label{eq:WSu}
\end{align}
From here on, we write $\SCC:=\SCC_k$.

\vspace{10pt}
\noindent\emph{Step 2: SCC adjoint rate matrix.}
The SCC block $\Wscc$ is irreducible and has nonnegative off-diagonal entries.
By the Perron--Frobenius theorem, $\Wscc$ has a simple eigenvalue
$\perrev\in\RR$ with largest real part and an associated right eigenvector
$\perron\in\RR_{>0}^{\vert \SCC\vert}$, such that $\Wscc\perron =
\perrev\perron$.
A straightforward application of Gershgorin's Circle Theorem implies $\perrev\le
0$.

Using the diagonal matrix $D=\operatorname{diag}(\perron)$, we define an adjoint
matrix,
\begin{align}
\Wadj \;:=\; \Dperron{\Wscc}^{\top} \Dperron^{-1} - \perrev I\,,
\label{eq:Wstar}
\end{align}
so that $\Wadj_{\xxyy}=\Wscc_{\yyxx}a_\xx/a_\yy$ for $\xx\ne\yy$.
Observe that $\Wadj$ is a Markovian generator on $\SCC$: it has nonnegative
off-diagonal entries and all column sums vanish:
\[
\sum_{\yy\in\SCC}\Wadj_{\yyxx}
=\Big( \sum_{\yy\in\SCC}\frac{\perxx_\yy}{\perxx_\xx}\Wscc_{\xxyy}\Big)%
-\perrev=0\,.
\]
In addition, the vector $\evecAdj:=\Dperron^{-1}\rvecS$ satisfies
\[
({\Wadj})^{\top}
 \evecAdj= {\Dperron}^{-1}{\Wscc} \rvecS -\perrev \evecAdj=
(\lambda-\perrev)\evecAdj \,,
\]
thus it is a {left} eigenvector of $\Wadj$ with eigenvalue $\lambda-\perrev$.

\vspace{10pt}
\noindent\emph{Step 3: Reduction to left-eigenvector case.}
We then apply the left-eigenvector case (established above) to $\Wadj$ with
nonreal eigenvalue $\adjLambda:=\lambda-\perrev$ and left eigenvector
$\evecAdj=\Dperron^{-1}\rvecS$. We write the maximum escape rate in $\Wadj$ as
\begin{align}
\label{eq:adjmaxratedef}\adjMaxrate:=\max_{\xx\in
\SCC}(-\Wadj_{\xx\xx})=\perrev+\max_{\xx\in \SCC}(-\W_{\xx\xx})\le
\perrev+\maxrate\,.
\end{align}
These results imply that $\imLadj=\im \lambda$ and
\begin{align}
\begin{gathered}
0<\minusReLadj = \minusReL + \perrev \le \minusReL\\
0<\dualReLadj \le 2\maxrate+\perrev+\reL\le \dualReL
\end{gathered}
\label{eq:vdsf3}
\end{align}
where we used $-2\adjMaxrate < \reLadj< 0$ (from Gershgorin's Circle Theorem)
and $\perrev\le 0$.

We indicate a cycle in $\Wadj$ as $\ccc$, and its affinity in $\Wadj$ as
$\adjAff_\ccc$. The left-eigenvector case of
Theorem~\ref{thm:multicyclic-main} implies that for each $\mixc\in[0,1]$, there
is an admissible cycle $\ccc$ with winding number
$\mccc\equiv\windFunc_\ccc(\evecAdj)\in\{1,\dots,\lceil n_\ccc/2\rceil-1\}$ such
that
\begin{align*}
\mixc\frac{|\imLadj|}{\minusReLadj}
\tan\!\frac{\pi \mccc}{n_\ccc}
+
(1-\mixc)\frac{|\imLadj|}{\dualReLadj}
\cot\!\frac{\pi \mccc}{n_\ccc}
\le \tanh\!\frac{\adjAff_\ccc}{2n_\ccc}.
\end{align*}
Using $\imLadj=\im \lambda$ and~\eqref{eq:vdsf3}, we can weaken this bound as
\begin{align*}
\mixc\frac{|\imL|}{\minusReL}
\tan\!\frac{\pi \mccc}{n_\ccc}
+
(1-\mixc)\frac{|\imL|}{\dualReL}
\cot\!\frac{\pi \mccc}{n_\ccc}
\le \tanh\!\frac{\adjAff_\ccc}{2n_\ccc}.
\end{align*}

Each cycle $\ccc = (\xx_1 , \xx_2, \dots, \xx_{n_\ccc} ,\xx_1)$ in $\Wadj$ can
be mapped to a unique cycle $c=(\xx_1, \xx_{n_\ccc}, \dots, \xx_2,\xx_1)$ in
$\W$ that has reverse order. The two cycles have the same size, $n_\ccc=n_c$. We
next show that they have the same affinity. If $\ccc$ contains an irreversible
edge, then so does $c$ and so $\adjAff_\ccc=\Aff_c=\infty$. Otherwise, both
cycles are reversible, and we have
\begin{align*}
\adjAff_\ccc
&=\sum_{\mathclap{\revedge\in c}}\ln\frac{\Wadj_{\yyxx}}{\Wadj_{\xxyy}}
=\sum_{\mathclap{\revedge\in
c}}\ln\frac{\Wscc_{\xxyy}(a_\yy/a_\xx)}{\Wscc_{\yyxx}(a_\xx/a_\yy)}
=\sum_{\mathclap{\revedge\in c}}\ln\frac{\Wscc_{\xxyy}}{\Wscc_{\yyxx}}=\Aff_c\,.
\end{align*}
Here $\Wscc_{\xxyy}=\W_{\xxyy}$ along $c$, and the logarithms of the factors
$a_\yy/a_\xx$ telescope around the cycle.

Finally, we show that the winding numbers $\windFunc_\ccc(\evecAdj)$ and
$\windFunc_c(\bm u)$ agree. Each edge $\yxedge\in\ccc$ corresponds to the
reversed edge $\revedge\in c$. Since
$\evecAdjxx_\yy/\evecAdjxx_\xx=(a_\xx/a_\yy)(\rvecSxx_\yy/\rvecSxx_\xx)$ with
$a_\xx/a_\yy>0$, the unit phase vectors on these two edges satisfy
\[
e^{\ii\Arg(\evecAdjxx_\yy/\evecAdjxx_\xx)}
=e^{-\ii\Arg(\rvecSxx_\xx/\rvecSxx_\yy)}\,.
\]
The negated phases still lie in a common open interval of length $\pi$, so $c$
is admissible for $\rvecS$. Let $\tilde\branchParam$ and $\branchParam^\SCC$ be
the branch parameters from Eq.~\eqref{eq:branchParamapp} for $(\evecAdj,\ccc)$
and $(\rvecS,c)$, respectively. The arguments of the average phase vectors obey
$\tilde\branchParam=-\branchParam^\SCC$ $(\bmod\,2\pi)$.

The corresponding lifted phases therefore differ by a minus sign and a common
integer multiple of $2\pi$. Reindexing by the reversed cycle, the sum in the
winding-number definition~\eqref{eq:winding-multicycle} becomes
\begin{multline*}
\sum_{\yxedge\in\ccc}\arg_{\tilde\branchParam}\frac{\evecAdjxx_\yy}%
{\evecAdjxx_\xx}
\equiv -\sum_{\yxedge\in
c}\arg_{\branchParam^\SCC}\frac{\rvecSxx_\yy}{\rvecSxx_\xx}
\quad(\bmod\,2\pi n_c).
\end{multline*}
The wrapping $|m|_{n_c}$ is invariant under sign change and addition of
multiples of $n_c$. Since $c$ lies entirely inside $\SCC$, we have
\begin{align}
\windFunc_\ccc(\evecAdj)=\windFunc_c(\rvecS)=\windFunc_c(\bm u) \,.
\label{eq:windAgree}
\end{align}

Using the affinity and winding identities above, this bound gives a cycle $c$ of
$\W$ with winding number $\winding_c\in\{1,\dots,\lceil n_c/2\rceil-1\}$ such
that
\begin{align*}
\mixc\frac{|\imL|}{\minusReL}
\tan\!\frac{\pi \winding_c}{n_c}
+
(1-\mixc)\frac{|\imL|}{\dualReL}
\cot\!\frac{\pi \winding_c}{n_c}
\le \tanh\!\frac{\Aff_c}{2n_c}.
\end{align*}
This proves the theorem when $\bm u$ is a right eigenvector.

\section{CYCLE EXTRACTION}
\label{app:cycle-extraction}

We introduce a few concepts from graph theory. Recall that a \emph{strongly connected
component} (SCC) of a directed graph is a maximal subset of
vertices such that every vertex is reachable from every other vertex by a
directed path. A \emph{sink SCC} is an SCC with no edges that leave the
component. Every nonempty finite directed graph has at least one nonempty sink
SCC: if one
contracts each SCC to a single vertex, the resulting graph is finite and
acyclic, so it has a vertex with no outgoing edges.

\begin{theorem}
\label{thm:witness-cycle-m4}
Let \(G=(V,E,\bm f)\) be a finite directed graph with signed edge weights
\(\bm f\in\RR^{|E|}\) and no self-loops.
Assume that:
\begin{enumerate}[itemsep=0pt]
\item The vertex set is nonempty: \(V\ne\emptyset\).
\item Every vertex $\xx\in V$ has nonnegative total outgoing weight:
$\sum_{\yy\in V:\yxedge\in E} \fyxedge \ge 0$.

\item Every vertex $\xx\in V$ has an outgoing edge $\yxedge\in E$ to some $\yy
\in V$.
\item Every nonnegative reversible edge has negative reverse weight:
if \(\yxedge\in E\), \(\fyxedge\ge0\), and \(\revedge\in E\), then
\(\fxyedge<0\).
\end{enumerate}
Then, \(G\) contains a directed simple cycle \(c\) that satisfies one of the
following:
\begin{enumerate}[itemsep=0pt]
\item[A.] All edges $\yxedge \in c$ are nonnegative (\(\fyxedge \ge 0\)) and at
least one is irreversible, \(\revedge\notin E\);
\item[B.] All edges $\yxedge \in c$ are positive and reversible
(\(\fyxedge>0,\revedge\in E\)) and satisfy the inequality
\[
\sum_{\yxedge\in c}\ln\frac{\fyxedge}{-\fxyedge} \ge 0\,.
\]
\label{it:witness-defect-m4}
\end{enumerate}
\end{theorem}

\begin{proof}

We call an edge \(\yxedge\in E\) \emph{nonnegative} if \(\fyxedge\ge0\).

\newcommand{\Gpos}{G^+}
\newcommand{\GposSCC}{\Gpos_\sinkSCC}
\newcommand{\GSCC}{G_\sinkSCC}

Let \(\Gpos\) denote the directed subgraph of \(G\) supported on
nonnegative edges.
Assumptions~2 and~3 imply that every vertex has a nonnegative outgoing
edge. Indeed, if all outgoing edges from some vertex had strictly negative
weight, their total outgoing weight would be strictly negative.
Thus, every vertex has an outgoing edge in
\(\Gpos\), so $\Gpos$ is nonempty.

Choose a sink SCC \(\sinkSCC\) of
\(\Gpos\), and let \(\GposSCC\) indicate the subgraph of \(\Gpos\) induced by
\(\sinkSCC\).
Observe that every vertex in \(\sinkSCC\) has an outgoing edge in
\(\Gpos\), and no such edge can leave \(\sinkSCC\). Therefore,
\(\GposSCC\) is nonempty and strongly connected,
thus it contains at least one cycle. We emphasize that, due to
Assumption~4 regarding the signed weights, cycles in \(\GposSCC\) cannot be
trivial ``roundtrips'' (\(\xx \to \yy \to \xx\)).

If $\GposSCC$ contains a cycle with an irreversible edge
$\yxedge$, then Claim~A holds immediately. Thus, it remains to consider the case
where no cycle in $\GposSCC$
has an irreversible edge.
Since $\GposSCC$ is strongly connected, every edge lies on a cycle and is
therefore reversible. By Assumption~4, its reverse has strictly negative weight.

Let \(\GSCC\) be the subgraph of \(G\) induced by \(\sinkSCC\). For every
\(\xx\in\sinkSCC\), all edges in \(E\) that leave \(\sinkSCC\) have negative
weight; otherwise, they would be nonnegative edges leaving the sink SCC of
$\Gpos$, a contradiction.
Therefore,
\begin{equation}
\sum_{{\yy\in\sinkSCC:\yxedge\in E}} \fyxedge
\ge
\sum_{{\yy\in V:\yxedge\in E}} \fyxedge
\ge0 \,.
\label{eq:restricted-outgoing-m4}
\end{equation}

Next, we define a rate matrix on $\sinkSCC$ with entries $W_{\yyxx}=\fyxedge$
for $\yxedge\in E(\GposSCC)$ (other off-diagonal entries are zero) and
$W_{\xx\xx}=-\sum_{\yy\ne\xx}W_{\yyxx}$. The rate matrix has at least one
nonzero steady-state distribution $\bm\pi$. For any $\yy$ with $\pi_\yy>0$,
stationarity implies $\sum_{\xx:\xx\ne\yy}\pi_\xx
W_{\yyxx}=\pi_\yy\sum_{\xx:\xx\ne\yy}W_{\xx\yy}$.
Plugging in the definition of $W$ gives
\begin{align*}
\sum_{\xx:\,\yxedge\in E(\GposSCC)}\pi_\xx\fyxedge
=\pi_\yy\sum_{\xx:(\xx\leftarrow\yy)\in E(\GposSCC)}f_{\xx\leftarrow\yy}\,.
\end{align*}
By separating nonnegative and negative outgoing weights, we can rewrite the
right-hand side  as
\begin{align}
&\pi_\yy\sum_{\xx:\,(\xx\leftarrow\yy)\in
E(\GSCC)}f_{\xx\leftarrow\yy}+\pi_\yy\sum_{{\xx:(\xx\leftarrow\yy)\in
E(\GSCC),f_{\xx\leftarrow\yy}<0}}(-f_{\xx\leftarrow\yy})\nonumber \\
&\quad\ge\pi_\yy\sum_{\substack{\xx:(\xx\leftarrow\yy)\in
E(\GSCC)\\f_{\xx\leftarrow\yy}<0}}(-f_{\xx\leftarrow\yy})
\nonumber\\
&\quad\ge \pi_\yy\sum_{\xx:\,\yxedge\in E(\GposSCC)}(-\fxyedge).
\label{eq:positive-balance-m4}
\end{align}
The first inequality uses Eq.~\eqref{eq:restricted-outgoing-m4}, while the
second restricts the negative outgoing edges to reverses of nonnegative incoming
edges.

Combining gives
\begin{equation*}
\sum_{\xx:\,\yxedge\in E(\GposSCC)}(\pi_\xx\fyxedge + \pi_\yy \fxyedge) \geq 0.
\end{equation*}
Hence for at least one incoming edge $\yxedge\in E(\GposSCC)$,\begin{align}
\pi_\xx\fyxedge\ge \pi_\yy(-\fxyedge)>0,
\label{eq:predecessor-edge-m4}
\end{align}
where the final strict inequality holds because $\pi_\yy > 0$ and $\fxyedge <0$
by Assumption~4. Thus, any $\yy$ with $\pi_\yy > 0$ admits a predecessor $i$
with $\pi_\xx>0,\fyxedge>0$.

Starting from any vertex $\yy$ with $\pi_\yy>0$, repeatedly choose such a
predecessor. All selected vertices remain in the finite positive support of
$\bm\pi$. Stop at the first repeated vertex; the selected edges between its two
occurrences form a directed simple cycle $c$. Multiplying
Eq.~\eqref{eq:predecessor-edge-m4} around this cycle cancels the stationary
probabilities:
\begin{align}
\prod_{\yxedge\in c}\frac{\fyxedge}{-\fxyedge}
\ge\prod_{\yxedge\in c}\frac{\pi_\yy}{\pi_\xx}=1.
\label{eq:predecessor-cycle-m4}
\end{align}
All edges of $c$ are positive and reversible. Taking logarithms proves Claim~B.

\end{proof}

\section{CONVEXITY PROOF}
\label{app:convexity}

Here we prove the convexity of the function $J_{\sigma,\psi}$ used in the proof
of our multicyclic result, Theorem~\ref{thm:multicyclic-main}. As a corollary,
we obtain the strict convexity of the function $\PhaseFunc_\psi$ from
Eq.~\eqref{eq:Gdef}, which is used in the unicyclic derivations.

\begin{theorem}
\label{thm:local-convexity}
Let \(\psi\in(-\pi/2,\pi/2),\sigma >0\) with
\(\sigma\ge |\sin\psi|\). Define
\begin{align*}
J_{\sigma,\psi}(\varphi,\xlog)
&:=
\ln\frac{e^\xlog\sin(\varphi-2\psi)+\sigma}
{e^{-\xlog}\sin\varphi-\sigma}
\end{align*}
on the domain
\begin{align*}
\jDomSG
:=
\left\{(\varphi,\xlog):
0<\varphi<\pi,\quad
\xlog<\ln\frac{\sin\varphi}{\sigma}
\right\}.
\end{align*}
Then, the domain \(\jDomSG\) is convex, and
the function \(J_{\sigma,\psi}\) is well-defined and jointly convex on
\(\jDomSG\).
\end{theorem}

\begin{proof}
The boundary function
\(\varphi\mapsto \ln(\sin\varphi/\sigma)\) differs by a constant from
\(\ln\sin\varphi\) on \((0,\pi)\). Therefore it is strictly concave there, and
\(\jDomSG\) is the strict subgraph of a concave function, hence
convex.

To show that \(J_{\sigma,\psi}\) is well-defined on
\(\jDomSG\), we introduce the denominator and numerator
\begin{align*}
A&:=e^{-\xlog}\sin\varphi-\sigma\,,&
B&:=e^{\xlog}\sin(\varphi-2\psi)+\sigma\,.
\end{align*}
Since \(0<\varphi<\pi\), we have \(\sin\varphi>0\), and the domain condition
\(e^\xlog<\sin\varphi/\sigma\)
gives \(A>0\). We will also use the inequality
\begin{align}
Be^{-\xlog}-Ae^\xlog
&=
\sin(\varphi-2\psi)-\sin\varphi+\sigma(e^\xlog+e^{-\xlog})\nonumber \\
&=
-2\sin\psi\cos(\varphi-\psi)+\sigma(e^\xlog+e^{-\xlog})\nonumber \\
&\ge -2|\sin\psi|+2\sigma\ge0,
\label{eq:smNMineq}
\end{align}
where we used the assumption \(\sigma\ge |\sin\psi|\). This shows that $B\ge
Ae^{2\xlog}>0$; therefore, we may write $J_{\sigma,\psi}=\ln B-\ln A$.

For our derivative calculations, we write \(J=J_{\sigma,\psi}\). Direct
differentiation gives
\begin{align*}
\partial_{\xlog}^2 J
&=
\frac{\sigma e^\xlog\sin(\varphi-2\psi)}{B^2}
+
\frac{\sigma e^{-\xlog}\sin\varphi}{A^2},\\
\partial_{\varphi}\partial_{\xlog}J
&=
\frac{\sigma e^\xlog\cos(\varphi-2\psi)}{B^2}
-
\frac{\sigma e^{-\xlog}\cos\varphi}{A^2},\\
\partial_{\varphi}^2 J
&=
-\frac{\sigma e^\xlog\sin(\varphi-2\psi)+e^{2\xlog}}{B^2}
+
\frac{e^{-2\xlog}-\sigma e^{-\xlog}\sin\varphi}{A^2}.
\end{align*}

We prove convexity by showing that the Hessian matrix $\nabla^2 J$
is positive semidefinite. Since it is a symmetric \(2\times2\) matrix, it is
enough to show that \(\operatorname{tr}\nabla^2 J\ge0\) and \(\det \nabla^2
J\ge0\). First,
\begin{align}
\operatorname{tr} \nabla^2 J
=
\frac{e^{-2\xlog}}{A^2}
-
\frac{e^{2\xlog}}{B^2}
\ge 0 ,
\label{eq:trH}
\end{align}
as follows from ${e^{-\xlog}}/{A}\ge {e^{\xlog}}/{B} > 0$, itself a consequence
of Eq.~\eqref{eq:smNMineq}.
Second, the determinant admits the sum-of-squares decomposition
\begin{align}
&\det \nabla^2 J=\frac{1}{A^3B^3}\Bigg[(\sigma^2+\cos^2\psi)\sin^2(\varphi-\psi)
\nonumber\\
&\quad\quad{}\times
\Bigl[\sigma(e^\xlog+e^{-\xlog})
-2\sin\psi\cos(\varphi-\psi)\Bigr]^2
\nonumber\\
&\quad
+(\sigma^2-\sin^2\psi)\cos^2(\varphi-\psi)
\nonumber\\
&\quad\quad{}\times
\Bigl[\sigma(e^\xlog-e^{-\xlog})
-2\cos\psi\sin(\varphi-\psi)\Bigr]^2\Bigg] \ge 0,
\label{eq:detSOS}
\end{align}
where we used that \(\sigma^2-\sin^2\psi\ge0\).
\end{proof}

The function $\PhaseFunc_\psi$ used in the unicyclic derivations is a special
case of $J_{\sigma,\psi}$.

\begin{corollary}
\label{cor:K-convexity}
For $\psi\in(0,\pi/2)$, the function $\PhaseFunc_\psi$ in Eq.~\eqref{eq:Gdef} is
well-defined and strictly convex on the convex domain
\begin{align*}
\Lambda_\psi
:=
\left\{(\phase,\xlog):
0<\phase<\pi-\psi,\quad
\xlog<\ln\frac{\sin(\phase+\psi)}{\sin\psi}
\right\}.
\end{align*}
\end{corollary}

\begin{proof}
Comparing Eq.~\eqref{eq:Gdef} with the definition of $J_{\sigma,\psi}$ gives
$\PhaseFunc_\psi(\phase,\xlog)=J_{\sin\psi,\psi}(\phase+\psi,\xlog)$, where
$\sigma=\sin\psi=|\sin\psi|>0$ satisfies the hypotheses of
Theorem~\ref{thm:local-convexity}. Under the affine map
$(\phase,\xlog)\mapsto(\phase+\psi,\xlog)$, the domain $\Lambda_\psi$ is the
preimage of $\jDom_{\sin\psi,\psi}$ intersected with the half-plane
$\{\phase>0\}$; hence $\Lambda_\psi$ is convex and is mapped into
$\jDom_{\sin\psi,\psi}$. Therefore, $\PhaseFunc_\psi$ is well-defined and convex
on $\Lambda_\psi$ by Theorem~\ref{thm:local-convexity}.

For strict convexity, note that the map
$(\phase,\xlog)\mapsto(\phase+\psi,\xlog)$ is a translation, so the Hessians of
$\PhaseFunc_\psi$ and $J_{\sin\psi,\psi}$ coincide at corresponding points:
\begin{align*}
\nabla^2\PhaseFunc_\psi(\phase,\xlog)=\nabla^2
J_{\sin\psi,\psi}(\phase+\psi,\xlog)\,.
\end{align*}
Hence $\operatorname{tr}\nabla^2\PhaseFunc_\psi\ge 0$ by Eq.~\eqref{eq:trH}.
Moreover, setting $\sigma=\sin\psi$ and $\varphi=\phase+\psi$ in
Eq.~\eqref{eq:detSOS} gives
\begin{align*}
\det \nabla^2
\PhaseFunc_\psi(\phase,\xlog)=%
\frac{\sin^2\phase\,\bigl[2\sin\psi(\cosh\xlog-\cos\phase)\bigr]^2}{A^3B^3}>0\,.
\end{align*}
Positivity follows from $\sin\phase>0$ and $\cosh\xlog\ge 1>\cos\phase$ for
$\phase\in(0,\pi-\psi)$. A symmetric $2\times 2$ matrix with nonnegative trace
and positive determinant is positive definite. Thus, the Hessian of
$\PhaseFunc_\psi$ is positive definite on $\Lambda_\psi$, so $\PhaseFunc_\psi$
is strictly convex there.
\end{proof}

\section{LOCALIZATION REGIONS}
\label{app:convex-hull}

Here, we derive the eigenvalue localization result,
Theorem~\ref{thm:multicyclic-hull}. As in the main text, we use $\mathcal W :=
\prod_{c\in\CCCpos} \big\{1,\dots,\lceil n_c/2\rceil -1\big\}$ to denote the set
of winding configurations.
For each winding configuration
\(\bm\windA=(\windA_1,\dots,\windA_{\vert \CCCpos\vert})\in\mathcal W
\), we define
\begin{align*}
\Omega(\bm\windA):=\Conv\Big(\{-2\maxrate,0\}\cup\bigcup_{c\in\CCCpos}%
\{\lUnif_{\Aff_c,n_c,\windA_c},\lUnifConj_{\Aff_c,n_c,\windA_c}\}\Big)\,,
\end{align*}
where $\lUnif_{\Aff,n,\windA}$ is the uniform $n$-cycle eigenvalue with winding
number $\windA$, escape rate $\maxrate$, and affinity $\Aff$,
\begin{align}
\lUnif_{\Aff,n,\windA} := \maxrate\bigg(\cos\frac{2\pi
\windA}{n}-1+\ii\tanh\frac{\Aff}{2n}\sin\frac{2\pi \windA}{n}\bigg)\,,
\label{eq:appunifeig}
\end{align}
and $\lUnifConj_{\Aff,n,\windA}$  is its complex conjugate.

We first prove the following lemma, which characterizes the upper envelope of a
fixed configuration-resolved convex hull.

\newcommand\hullFunc{U}
\begin{lemma}
\label{lem:hull-envelope-main}
Assume \(\CCCpos\ne\varnothing\). Then, for any \(\bm\windA\in\mathcal W\),
\[
\Omega(\bm\windA)=\bigl\{x+\ii y\in\mathbb{C}:-2\maxrate\le x\le
0,\ |y|\le\hullFunc_{\bm \windA}(x)\bigr\},
\]
where \(\hullFunc_{\bm \windA}(-2\maxrate)=\hullFunc_{\bm \windA}(0)=0\) and,
for \(x\in(-2\maxrate,0)\),
\begin{align}
\hullFunc_{\bm \windA}(x):=
\min_{\mixc\in[0,1]}\,\max_{c\in\CCCpos}\,
\frac{\tanh(\Aff_c/2n_c)}
{\frac{\mixc}{-x}\tan\!\frac{\pi \windA_c}{n_c}
+\frac{1-\mixc}{2\maxrate+x}\cot\!\frac{\pi \windA_c}{n_c}}.
\label{eq:phi-multicyclic}
\end{align}
\end{lemma}

\begin{proof}
Since the convex hull $\Omega(\bm\windA)$ is symmetric about the real axis, we
only need to
compute its upper envelope and then reflect it. Moreover, the upper envelope of
a convex hull is the pointwise minimum over all affine functions that bound the
generating points of the hull: $\lUnif_{\Aff_c,n_c,\windA_c}$ and the
real-valued endpoints $-2\maxrate,0$. (An affine function $h:\RR \to \RR$ is any
function that can be expressed in the form $h(x)=c_1 x + c_2$.)

For each cycle \(c\in\CCCpos\), we may express the real and imaginary parts of
the generating point as
\begin{align*}
x_c &:=\re \lUnif_{\Aff_c,n_c,\windA_c}
=\maxrate\cos\frac{2\pi \windA_c}{n_c}-\maxrate\,,\\
y_c&:=\im \lUnif_{\Aff_c,n_c,\windA_c}= \maxrate
\tanh\frac{\Aff_c}{2n_c}\sin\frac{2\pi \windA_c}{n_c}\,.
\end{align*}
Then, we formally write the upper envelope \(\hullFunc_{\bm \windA}(x)\) as
\begin{align}
\hullFunc_{\bm \windA}(x)
\label{eq:pointwise-envelope}
&=
\min_{\text{affine\,}h:\RR\to\RR} h(x)\\
&\text{where}\,\, h(-2\maxrate)\ge0,\ h(0)\ge0, h(x_c)\ge y_c \ \forall
c\in\CCCpos\,.\nonumber
\end{align}
We solve this minimization at each fixed $x\in(-2\maxrate,0)$; the endpoints
$x=-2\maxrate$ and $x=0$ follow by continuity.
First, without loss of generality, we parameterize each affine $h$ as
\begin{align}\label{eq:aff0}
h(x)=
-\frac{x}{2\maxrate}h(-2\maxrate)
+\frac{2\maxrate+x}{2\maxrate}h(0)\,.
\end{align}
Note that $h(x)>0$ for $x\in(-2\maxrate,0)$, as follows from the constraints $h(-2\maxrate)\ge0, h(0)\ge0$ and $h(x_c)\ge y_c >0$ for at least one cycle. Therefore, 
the constraint $h(x_c)\ge y_c$ for all $c\in\CCCpos$ can be written as
\begin{align}
1&\ge \max_{c \in \CCCpos}\frac{y_c}{-\frac{x_c}{2\maxrate}h(-2\maxrate)
+\frac{2\maxrate+x_c}{2\maxrate}h(0)}\,.
\label{eq:consapp0}
\end{align}
Next, for this fixed $x$, we index the possible affine functions $h$ in terms of
the parameter
\begin{align}
\label{eq:zvx2}
\mixc:=-\frac{x}{2\maxrate}\frac{ h(-2\maxrate)}{h(x)}.
\end{align}
We may solve for $h(-2\maxrate)$ and $h(0)$, also using Eq.~\eqref{eq:aff0}, as
\begin{align}
\label{eq:zvx3}
h(-2\maxrate)=-\frac{2\maxrate \mixc h(x)}{ x}\,,\qquad h(0)=\frac{2\maxrate
(1-\mixc) h(x)}{2\maxrate + x}.
\end{align}
Therefore, we may rewrite Eq.~\eqref{eq:consapp0} as
\begin{align*}
h(x)&\ge \max_{c \in \CCCpos}\frac{y_c}{\frac{\mixc}{x}x_c+
\frac{1-\mixc}{2\maxrate+x}(2\maxrate+x_c)}\,.
\end{align*}

Next, we express the upper envelope by finding the pointwise minimum $h$, i.e.,
by minimizing over $\tau$ at each $x$. However, given $x\in(-2\maxrate,0)$ and
Eq.~\eqref{eq:zvx3}, the constraint $h(-2\maxrate)\ge 0$ requires that $\mixc\ge
0$. Similarly, given Eq.~\eqref{eq:zvx3}, the constraint $ h(0)\ge 0$ requires
that $\mixc\le 1$. Using these bounds, we have
\begin{align}
\hullFunc_{\bm \windA}(x)=
\min_{\mixc\in[0,1]}\,
\max_{c\in\CCCpos}
\frac{y_c}
{\frac{\mixc}{x}x_c+
\frac{1-\mixc}{2\maxrate+x}(2\maxrate+x_c)}\,.\label{eq:zccvxs3}
\end{align}

Using the definitions of $(x_c,y_c)$ and half-angle
identities, we have
\[
\frac{-x_c}{y_c}=\frac{\tan({\pi \windA_c}/{n_c})}{\tanh(\Aff_c/2n_c)}\,,\qquad
\frac{2\maxrate+x_c}{y_c}=
\frac{\cot({\pi \windA_c}/{n_c})}{\tanh(\Aff_c/2n_c)}\,.
\]
Substituting into Eq.~\eqref{eq:zccvxs3} gives
Eq.~\eqref{eq:phi-multicyclic}.
\end{proof}

Next, we derive Theorem~\ref{thm:multicyclic-hull}.

\begin{reptheoremnum}{thm:multicyclic-hull}
Every eigenvalue \(\lambda\)
belongs to $\bigcup_{\bm \windA\in\mathcal W}\Omega(\bm\windA)$.
\end{reptheoremnum}
\begin{proof}
By Gershgorin's Circle Theorem, real eigenvalues lie in
\([-2\maxrate,0]\subset\Omega(\bm\windA)\) for every
\(\bm\windA\in\mathcal W\).
Also, if \(\CCCpos=\varnothing\), then Theorem~\ref{thm:multicyclic-main} rules
out
nonreal eigenvalues.
Therefore, we assume \(\CCCpos\ne\varnothing\) and focus our attention on
nonreal eigenvalues.

Consider a nonreal eigenvalue \(\lambda\) and associated eigenvector
\(\bm u\).
For every \(\mixc\in[0,1]\), Theorem~\ref{thm:multicyclic-main}
provides a certifying admissible cycle \(c\in\CCCpos\) such that
$\windFunc_c(\bm u)\in \{1,\dots,\lceil n_c /2\rceil - 1 \}$ and
\begin{align}
\absImL\le
\frac{\tanh(\Aff_c/2n_c)}
{\frac{\mixc}{\minusReL}\tan\frac{\pi \windFunc_c(\bm u)}{n_c}+
\frac{1-\mixc}{\dualReL}\cot\frac{\pi \windFunc_c(\bm u)}{n_c}}\,.
\label{eq:hull-thm-bound0}
\end{align}
For cycles that certify Theorem~\ref{thm:multicyclic-main} for at least one
\(\mixc\), define \(\windA_c:=\windFunc_c(\bm u)\); for all other cycles, set
\(\windA_c=1\).
This arbitrary assignment is harmless because, for each
\(\mixc\), the maximization below includes at least one certifying cycle.

Each $\windA_c\in \{1,\dots,\lceil n_c /2\rceil - 1 \}$, therefore
\(\bm\windA\in\mathcal W\).
Then we weaken the bound~\eqref{eq:hull-thm-bound0} by taking the maximum over
all positive-affinity cycles:
\begin{align}
\absImL\le\max_{c\in\CCCpos}\,
\frac{\tanh(\Aff_c/2n_c)}
{\frac{\mixc}{\minusReL}\tan\!\frac{\pi \windA_c }{n_c}+
\frac{1-\mixc}{\dualReL}\cot\!\frac{\pi \windA_c}{n_c}}\,.
\label{eq:hull-thm-bound}
\end{align}
Since this inequality holds for every \(\mixc\in[0,1]\), we may minimize the
right-hand
side over \(\mixc\) to obtain \(\absImL\le\hullFunc_{\bm \windA}(\reL)\). By
Lemma~\ref{lem:hull-envelope-main}, \(\lambda\in\Omega(\bm\windA)\).
Since \(\bm\windA\in\mathcal W\), this also implies that
$\lambda\in\bigcup_{\bm\windA\in\mathcal W}\Omega(\bm\windA)$.
\end{proof}

We now prove Theorem~\ref{thm:topology-free-region}, our topology-free
localization result. Here, we use \(F_n\) to denote the Farey sequence of order
\(n\) and define $\mathfrak F^*:=\max_{c\in\CCC}\tanh({\Aff_c}/{2n_c})$ as the
maximum normalized affinity ($\mathfrak F^*=0$ if no cycles exist).

\begin{reptheoremnum}{thm:topology-free-region}
Every eigenvalue \(\lambda\) belongs to
\begin{align*}
\Conv\Big\{
\maxrate\big[
\cos(2\pi q)-1+
\ii\,\mathfrak F^*\sin(2\pi q)
\big]
: q\in F_n
\Big\}.
\end{align*}
\end{reptheoremnum}

\begin{proof}
Fix a winding configuration $\bm\windA\in\mathcal W$.
For each $c\in\CCCpos$, the fraction $q=\windA_c/n_c$ belongs to
$F_n\cap(0,1/2)$ because $n_c\le n$.
Since $1-q\in F_n$ as well, the convex hull in the statement contains both
points
\[
\maxrate\big[\cos(2\pi q)-1\pm\ii\,\mathfrak F^*\sin(2\pi q)\big].
\]
Since $0\le\tanh(\Aff_c/2n_c)\le\mathfrak F^*$, the segment joining these two
points contains the points $\lUnif_{\Aff_c,n_c,\windA_c}$ and
$\lUnifConj_{\Aff_c,n_c,\windA_c}$.
The convex hull stated in this theorem also contains $0$ and $-2\maxrate$, since
these correspond to $q=0$ and $q=1/2$.
Thus, this convex hull contains every generating point of $\Omega(\bm\windA)$,
and hence $\Omega(\bm\windA)$ itself, for every $\bm\windA\in\mathcal W$.
Theorem~\ref{thm:multicyclic-hull} then gives the result.

\end{proof}

\section{MULTICYCLIC WINDING--INDEX COUNTEREXAMPLE}
\label{app:winding-index-counterexample}

In unicyclic systems, Theorem~\ref{thm:unicyclic-bounds} can be equivalently
expressed either in terms of the winding number $\winding$ or the ordered
eigenvalue index $k$, since Theorem~\ref{thm:ordered-winding} shows that
$\winding=\lfloor k/2\rfloor$. This raises the question of whether a similar
relationship holds between the eigenvalue index and winding numbers in the
multicyclic bound, Theorem~\ref{thm:multicyclic-main}. Here we show using a
numerical counterexample that straightforward generalizations of this
relationship cannot hold in general.

Our example is a ten-state ring with an additional bidirectional edge. The
ordered eigenvalues are as follows:
\[
\begin{aligned}
\lambda_1     &=0,                       & \lambda_{6,7}&\approx-23.71\pm7.14\,\ii,\\
\lambda_{2,3} &\approx-5.03\pm0.56\,\ii, & \lambda_{8,9}&\approx-35.48\pm6.31\,\ii,\\
\lambda_{4,5} &\approx-5.06\pm4.09\,\ii, & \lambda_{10} &\approx-40.67.
\end{aligned}
\]

The transition rates are shown below.

\begin{center}
\scalebox{0.7}{%
\begingroup
\fontfamily{phv}\fontsize{8.5}{9.5}\selectfont
\expandafter\let\expandafter\tenSMmeasurefont\the\font
\newcommand{\tenSMratepair}[5]{%
  \path let \p1=($(#2.center)-(#1.center)$),
    \n1={atan2(\y1,\x1)} in
    coordinate (ten-sm-edge-#1-#2-start) at ($(#1.center)+(\n1+90:0.45mm)$)
    coordinate (ten-sm-edge-#1-#2-end) at ($(#2.center)+(\n1+90:0.45mm)$)
    coordinate (ten-sm-edge-#2-#1-start) at ($(#2.center)+(\n1-90:0.45mm)$)
    coordinate (ten-sm-edge-#2-#1-end) at ($(#1.center)+(\n1-90:0.45mm)$);
  \draw[ten sm transition] (ten-sm-edge-#1-#2-start) -- (ten-sm-edge-#1-#2-end);
  \draw[ten sm transition] (ten-sm-edge-#2-#1-start) -- (ten-sm-edge-#2-#1-end);
  \expandafter\def\csname ten-sm-rate-#1-#2\endcsname{#3}%
  \expandafter\def\csname ten-sm-rate-#2-#1\endcsname{#4}%
  \expandafter\def\csname ten-sm-labelpos-#1-#2\endcsname{#5}%
  \expandafter\def\csname ten-sm-labelpos-#2-#1\endcsname{#5}%
}
\newcommand{\tenSMratelabel}[3]{%
  \begingroup
  \def\tenSMlabelpos{#3}%
  \ifx\tenSMlabelpos\empty
    \edef\tenSMlabelpos{\csname ten-sm-labelpos-#1-#2\endcsname}%
  \fi
  \setbox0=\hbox{\tenSMmeasurefont
    \csname ten-sm-rate-#1-#2\endcsname}%
  \dimen0=.5\wd0 \advance\dimen0 by .3pt
  \ifdim\dimen0<.975mm \dimen0=.975mm\fi
  \edef\tenSMhalfwidth{\the\dimen0}%
  \path let \p1=($(ten-sm-edge-#1-#2-end)-(ten-sm-edge-#1-#2-start)$),
    \n1={atan2(\y1,\x1)},
    \n2={(\tenSMhalfwidth)*abs(sin(\n1))+1.225mm*abs(cos(\n1))+0.25mm}
    in node[ten sm rate] at
      ($(ten-sm-edge-#1-#2-start)!\tenSMlabelpos!(ten-sm-edge-#1-#2-end)+(\n1+%
90:\n2)$)
      {\csname ten-sm-rate-#1-#2\endcsname};
  \endgroup
}
\begin{tikzpicture}[
  ten sm state/.style={circle,draw=none,fill=black!40,
    minimum size=4.5mm,inner sep=0pt,text=black,
    font=\fontfamily{phv}\fontsize{8.5}{9.5}\selectfont},
  ten sm transition/.style={-{Latex[length=1.8mm,width=1.2mm]},
    draw=black,line width=0.65pt,shorten <=2.45mm,shorten >=2.45mm,
    preaction={draw=white,line width=2pt}},
  ten sm rate/.style={fill=white,inner sep=0.3pt,outer sep=0pt,
    minimum width=1.95mm,minimum height=2.45mm,
    text=black,text height=6.1pt,text depth=0pt,
    font=\fontfamily{phv}\fontsize{8.5}{9.5}\selectfont},
  line cap=round,line join=round
]
  \foreach \i in {0,...,9}
    \coordinate (ten-sm-\i) at ({90-36*\i}:2.65cm);

  \tenSMratepair{ten-sm-0}{ten-sm-1}{5}{1}{0.50}
  \tenSMratepair{ten-sm-1}{ten-sm-2}{28}{1}{0.50}
  \tenSMratepair{ten-sm-2}{ten-sm-3}{3}{1}{0.50}
  \tenSMratepair{ten-sm-3}{ten-sm-4}{28}{1}{0.50}
  \tenSMratepair{ten-sm-4}{ten-sm-5}{5}{1}{0.50}
  \tenSMratepair{ten-sm-5}{ten-sm-6}{28}{1}{0.50}
  \tenSMratepair{ten-sm-6}{ten-sm-7}{28}{1}{0.50}
  \tenSMratepair{ten-sm-7}{ten-sm-8}{5}{1}{0.50}
  \tenSMratepair{ten-sm-8}{ten-sm-9}{8}{1}{0.50}
  \tenSMratepair{ten-sm-9}{ten-sm-0}{3}{1}{0.50}
  \tenSMratepair{ten-sm-8}{ten-sm-1}{28}{707/3000}{0.50}

  \foreach \i in {0,...,9}
    \node[ten sm state] at (ten-sm-\i) {};

  \foreach \i in {0,...,9} {
    \pgfmathtruncatemacro{\j}{mod(\i+1,10)}
    \tenSMratelabel{ten-sm-\i}{ten-sm-\j}{0.50}
    \tenSMratelabel{ten-sm-\j}{ten-sm-\i}{0.50}
  }
  \tenSMratelabel{ten-sm-8}{ten-sm-1}{0.50}
  \tenSMratelabel{ten-sm-1}{ten-sm-8}{0.35}
\end{tikzpicture}
\endgroup
}
\end{center}

There are three cycles up to reversal: $c_4$ is the upper loop with 4 states,
$c_8$ is the lower loop with 8 states, and $c_{10}$ is the outer ring with 10
states.
When traversed clockwise, the cycles have positive affinity $\Aff_4\approx
0.01$, $\Aff_8\approx 22.42$, and $\Aff_{10}\approx 22.43$. Their reversed
orientations have negative affinity and cannot certify.

Let us first consider the left eigenvector of the slowest mode $k=2$ (the
results for the right eigenvector and for the conjugate mode are exactly the
same). Cycle $c_8$ is not admissible  because the eigenvector edge-ratio phases
do not fall into an open semicircle. Cycle $c_4$ is admissible and has winding
number $\winding_4=1$; however, this cycle fails the inequality of
Theorem~\ref{thm:multicyclic-main} at both $\tau=0$ and $\tau=1$ (and hence
throughout $[0,1]$, since the left-hand side is linear in $\tau$). Therefore,
$c_4$ cannot be the certifying cycle. The unique certifying cycle $c_{10}$ has
$\winding_{10}=2>\lfloor k/2\rfloor=1$.

Next, we consider the left eigenvector for $k=4$, representing the
second-slowest conjugate pair (the results for the right eigenvector and for the
conjugate mode are exactly the same). Cycles $c_4$ and $c_{10}$ are not
admissible because their edge-ratio phases do not fit within an open semicircle.
Cycle $c_8$ is admissible and has winding number $\winding_8=1$. Thus, it is the
unique certifying cycle in Theorem~\ref{thm:multicyclic-main}, while having
$\winding_8=1<\lfloor k/2\rfloor=2$.

Consequently, neither  $\winding_c\ge \lfloor k/2\rfloor$ nor $\winding_c\le
\lfloor k/2\rfloor$ can hold in general for certifying cycles in
Theorem~\ref{thm:multicyclic-main}. We have also observed that, unlike the
unicyclic case, the winding numbers of the certifying cycles can decrease with
eigenvalue index.